\pdfoutput=1
\documentclass[10pt,pra,aps,superscriptaddress,nofootinbib,twocolumn,longbibliography]{revtex4-1}
\usepackage{amsmath,bbm}
\usepackage{amsthm}
\usepackage{amsmath}
\usepackage{latexsym}
\usepackage{amssymb}
\usepackage{float}
\usepackage{quantikz}
\usepackage{tikz}
\usepackage{graphicx}           
\usepackage{color}
\usepackage{xcolor}
\usepackage{mathpazo}
\usepackage{comment}
\usepackage{enumitem}
\usepackage{multirow}
\usepackage{setspace}
\usepackage{subcaption}
\usepackage[colorlinks=true,linkcolor=blue,citecolor=magenta,urlcolor=blue]{hyperref}
\usepackage{svg}

\newcommand{\be}{\begin{equation}}
\newcommand{\ee}{\end{equation}}
\newcommand{\bea}{\begin{eqnarray}}
\newcommand{\eea}{\end{eqnarray}}
\def\squareforqed{\hbox{\rlap{$\sqcap$}$\sqcup$}}
\def\qed{\ifmmode\squareforqed\else{\unskip\nobreak\hfil
\penalty50\hskip1em\null\nobreak\hfil\squareforqed
\parfillskip=0pt\finalhyphendemerits=0\endgraf}\fi}
\def\endenv{\ifmmode\;\else{\unskip\nobreak\hfil
\penalty50\hskip1em\null\nobreak\hfil\;
\parfillskip=0pt\finalhyphendemerits=0\endgraf}\fi}

\newcommand{\tr}{\text{Tr}}
\newcommand{\I}{\mathbbm{1}}

\newcommand{\la}{\langle}
\newcommand{\ra}{\rangle}

\makeatletter
\newtheorem*{rep@theorem}{\rep@title}
\newcommand{\newreptheorem}[2]{%
\newenvironment{rep#1}[1]{%
 \def\rep@title{#2 \ref{##1}}%
 \begin{rep@theorem}}%
 {\end{rep@theorem}}}
\makeatother

\newtheorem{thm}{Theorem}
\newreptheorem{thm}{Theorem}
\newtheorem{lemma}{Lemma}

\begin{document}

\title{Entanglement assisted communication complexity measured by distinguishability }


\author{Satyaki Manna}
\affiliation{Department of Physics, School of Basic Sciences, Indian Institute of Technology Bhubaneswar, Odisha 752050, India}
\author{Ankush Pandit}
\affiliation{Department of Physics, School of Basic Sciences, Indian Institute of Technology Bhubaneswar, Odisha 752050, India}
\author{Anubhav Chaturvedi}
\affiliation{Faculty of Applied Physics and Mathematics, Gda\'nsk University of Technology, Gabriela Narutowicza 11/12, 80-233 Gda\'nsk, Poland}
\affiliation{International Centre for Theory of Quantum Technologies (ICTQT), University of Gda\'nsk, 80-308 Gda\'nsk, Poland}
\author{Debashis Saha}
\affiliation{Department of Physics, School of Basic Sciences, Indian Institute of Technology Bhubaneswar, Odisha 752050, India}

\begin{abstract}
We investigate the quantum advantage that can arise in typical two-party communication scenarios, where the sender and the receiver are allowed to share prior correlations. Focusing on communication tasks constrained by the distinguishability of the sender's inputs, we demonstrate that entanglement-assisted communication with both classical and quantum message can outperform classical communication supplemented with shared randomness.
We begin by developing a general framework for communication tasks with pre-shared correlations. Within this framework, we establish an equivalence among entanglement-assisted classical communication, entanglement-assisted quantum communication, and quantum communication, showing that no hierarchy exists between these three paradigms. We then investigate the scenario where the receiver has no input and prove that no advantage can arise in this case. However, an advantage in the entanglement-assisted setting emerges once additional constraints are imposed on the dimension of the communicated message. This further highlights the superiority of entanglement-assisted classical communication over standard quantum communication. Then we demonstrate several tasks where the entanglement-assisted protocol using one-bit communication proves to be advantageous over classical communication.
Finally, by constructing an explicit class of communication tasks, we show that a non-maximally entangled states outperform the maximally entangled state as a pre-shared resource between the communicating parties.
\end{abstract}

\maketitle


\section{Introduction} \label{SEC I} 
Communication complexity is a fundamental topic in information science with applications in cryptography, query complexity, distributed computation and algorithms \cite{ccbook,Yao,DEWOLF}. Traditionally, communication tasks are analyzed by imposing an upper bound on the alphabet size of a classical message or the Hilbert space dimension of a quantum message. Numerous results have demonstrated that a lower-dimensional quantum system can accomplish certain tasks that would otherwise require a higher-dimensional classical system \cite{fingerp,saha2019,saha2023,Saha_2019,saha.062210,Chaturvedi}. A more recent perspective instead constrains the distinguishability of the communicated message \cite{Chaturvedi2020quantum,PRR_24,pandit2025,pandit2,Tavakoli2022}. Distinguishability is defined as the maximum average probability of correctly guessing the sender's input from the transmitted message \cite{Hellstrom,herzog,benett2,AYuKitaev_1997,piani,watrous,acin2,npj,duan,manna2026,manna2025}. It quantifies how much information about the sender's data is revealed through communication. A higher distinguishability requirement therefore, implies greater information conveyed about the sender's input and consequently a larger communication cost, whereas lower distinguishability indicates reduced communication requirements. While dimension provides only a discrete and partial characterization of communication resources, distinguishability is a continuous measure that offers a more complete description. Moreover, minimizing distinguishability becomes particularly important when the sender’s input must remain confidential. 

Entanglement is widely regarded as one of the most intriguing and distinctive features of quantum theory. One of its most prominent applications is its role as a nonlocal resource in Bell-type scenarios \cite{bell,ardehali,mermin,seevnick}, where spatially separated parties are not allowed to communicate yet can generate correlations that cannot be explained classically. In standard communication tasks, the role and potential advantage of entanglement have not been greatly explored. Recent studies have shown that the introduction of entanglement can enlarge the set of achievable correlations between a sender and a receiver in scenarios where communication is constrained by the dimension of the transmitted systems \cite{Marcin,Brukner,armin,tavakoli,hameedi,5550293,Frenkel2022,Vieira_2023,jef,Patra}.  
However, whether entanglement can offer a similar advantage in communication tasks constrained by distinguishability remains unexplored. This question is particularly relevant, as distinguishability constraint is defined relative to pre-shared correlations, and it provides a natural and unbiased framework for studying quantum communication complexity with various resources. 

 In a seminal work, Frenkel and Weiner~\cite{fw} showed that, in dimension-constrained communication tasks without any input at the receiver’s side, quantum communication offers no advantage over classical strategies. Despite this limitation, subsequent results by Frenkel and Weiner~\cite{Frenkel2022} and Vieira \textit{et al.}~\cite{Vieira_2023} revealed that entanglement-assisted classical communication can outperform classical protocols in dimension-limited scenarios. These developments naturally raise the question of whether a similar advantage persists in distinguishability-constrained communication tasks. The answer in this case appears to be negative. However, by introducing communication tasks constrained simultaneously by both the distinguishability of the sender’s inputs and the dimension of the communicated message, we show that such an advantage can, in fact, be recovered.

We begin by describing a general two-party communication task under the constraint on the distinguishability of the sender’s inputs. This notion has been previously developed in Refs.~\cite{Chaturvedi2020quantum, PRR_24}. We then evaluate the explicit form of this constraint in four communication scenarios: $(i)$ classical communication with shared randomness, $(ii)$ quantum communication, $(iii)$ entanglement-assisted classical communication, and $(iv)$ entanglement-assisted quantum communication.
Notably, in the first scenario, the constraint reduces to the same form as in classical communication without shared randomness, as discussed in Ref.~\cite{PRR_24}. One may therefore think of taking this classical constraint without shared randomness. In this case, we show that the communication complexity becomes trivial. The receiver can always succeed when the sender’s inputs are minimally distinguishable.
In this work, our primary focus is on quantifying the advantage of entanglement-assisted communication (both classical and quantum) over classical communication. We capture this advantage in two distinct ways. First, we consider the ratio between the distinguishability achievable in classical communication and that in entanglement-assisted classical (or entanglement-assisted quantum) communication, while fixing the success metric to be the same in both scenarios. Second, we fix the distinguishability in the respective communication settings and compare the corresponding success metrics by taking the ratio between the entanglement-assisted scenario and the classical one. In both approaches, a ratio exceeding $1$ signifies an advantage arising from entanglement assistance.

Turning to our main results, we first establish several useful lemmas that demonstrate the equivalence of three scenarios: quantum communication, entanglement-assisted classical communication, and entanglement-assisted quantum communication. This equivalence implies that no task can induce a hierarchical separation among these three protocols. In other words, advantage in one scenario indicates the same advantage in other two scenarios. Then we address the class of communication task where the receiver does not have input. We prove that all three protocols can not provide advantage with respect to classical communication in this case.
Then we introduce a different class of communication protocols that simultaneously constrain the distinguishability of the sender’s inputs and the dimension of the communicated message. Within this restricted framework, we show that entanglement-assisted classical communication and entanglement-assisted quantum communication provides an advantage with respect to not only classical communication but also standard quantum communication even if there is no input on receiver.

Next, we establish an equivalence between entanglement-assisted classical communication and quantum communication without pre-shared entanglement when the communicated message is restricted to one bit. In particular, we show that any qubit protocol in the latter scenario, where the receiver performs projective measurements, can be simulated within the former framework.
We then analyze three important classes of communication tasks: random access codes, equality problems defined via graphs, and pairwise distinguishability tasks. These tasks were shown in Ref.~\cite{PRR_24} to exhibit an advantage with respect to classical communication in qubit binary-output settings. By exploiting our equivalence relations, we demonstrate that any advantage observed in quantum communication can be translated into an advantage in entanglement-assisted classical communication, and consequently, into entanglement-assisted quantum communication as well. It is important to note that we allow infinite amount of messages in classical communication, where as entanglement assisted scenario is advantageous with one classical bit transmission.
Finally, we provide an explicit example of a communication task in which a pre-shared non-maximally entangled state yields a strictly greater advantage than a maximally entangled state.

The paper is organized as follows. In the next section, we explicitly formulate the distinguishability constraint and the corresponding success probability for four classes of communication scenarios. We then present several general results that establish connections among different types of communication protocols. Next two sections provide the advantageous entanglement assisted protocols over both classical and quantum communication when the dimension of communicated message is also restricted. Finally, we conclude by summarizing our main findings and outlining possible directions for future research.
\section{General Communication Task}
First, we introduce the notation $[K]$ to denote the set $\{1,\ldots,K\}$ for any positive integer $K$. In a one-way communication-complexity task, Alice (the sender) receives an input variable $x \in [X]$ drawn according to a distribution $\{p_x\}$, where for the uniform case $p_x=1/X$. In each round, conditioned on the value of $x$, Alice sends a message---classical or quantum---to Bob (the receiver). Bob is independently given an input variable $y \in [Y]$, and upon receiving Alice’s message, he produces an output $z \in [Z]$. Repeating the task many times yields empirical conditional probabilities $p(z|x,y)$.

The goal is to maximize a success metric of the form
\be\label{SGen}
\mathcal{S} = \sum_{x,y,z} c(x,y,z)\, p(z|x,y),
\ee
where $c(x,y,z)$ are weights associated with the respective probabilities.

The communication is constrained by the distinguishability of Alice’s inputs. This distinguishability is defined as the optimal average probability of correctly guessing $x$ from the communicated message, optimized over all admissible measurements:
\be\label{pD1}
\mathcal{D} = \max_{M}\left\{ \sum_x p_x\, p(z=x|x,M) \right\}.
\ee
If $\mathcal{D}=1$, Alice could perfectly encode and transmit the input $x$ itself (for example, using an $N$-dimensional classical system), trivializing the task. Hence, the task is meaningful only when $\mathcal{D}<1$. It is important to note that the distinguishability $\mathcal{D}$ is considered with respect to Bob (the receiver) with whom Alice may share classical or quantum correlations. We can define $\mathcal{D}$ in two distinct ways. First, it may quantify distinguishability with respect to both the communicated message and the pre-shared resources shared between the parties. Alternatively, one may define a notion of distinguishability that depends solely on the communicated message, disregarding any pre-shared resources.

Typically, the success metric in communication-complexity problems corresponds to correctly guessing a function of the inputs $x$ and $y$, leading to the form
\begin{equation}\label{Sgf}
\mathcal{S} = \sum_{x,y} c(x,y)\, p\bigl(z = f(x,y)\,\big|\, x,y\bigr).
\end{equation}
The expression in \eqref{Sgf} is thus a specific instance of the general form in \eqref{SGen}, where the output $z$ is required to equal a particular function $f(x,y)$. In this work, we focus on a family of tasks whose success metric is of the form \eqref{Sgf}. We next describe how such a communication task is instantiated in different frameworks and clarify the notion of \emph{quantum advantage} within each framework.

\subsection{Classical Communication with shared randomness (CC)}

 Alice transmits a $d$-labelled message $m'$ depending on the input $x$ and shared randomness $\lambda$, where $\sum_\lambda p(\lambda)=1$. Any encoding strategy is described by probability distributions of sending $m'$ given input $x$ and shared randomness $\lambda$, $\{p_e(m'|x,\lambda)\}$, where $\sum_{m'} p_e(m'|x,\lambda)=1$ for all $x$. It is important to note that there is no restriction on the dimension $d$, implying $m'$ can take an arbitrarily large number of distinct values.\\

 The decoding strategy for providing Bob's output is defined by the probability distributions over output $z$ given the message $m'$, Bob's inputs $y$ and shared randomness $\lambda$ denoted as $\{p_d(z|y,m',\lambda)\}$, where $\sum_z p_d(z|y,m',\lambda)=1$ for all $m',y$. For this kind of classical encoding and decoding, the resulting conditional probability is expressed as follows:
\be \label{pC}
p(z|x,y) = \sum_{m',\lambda }p(\lambda)p_e(m'|x,\lambda) p_d(z|y,m',\lambda).
\ee  
The distinguishability constraint on the sender’s inputs can be formulated in two ways. First, one may define distinguishability by taking into account both the communicated message and the shared randomness. Given an encoding, the distinguishability of sender's inputs in this scenario,
\bea\label{5}
\mathcal{D}_C &=& \max_{\{p_d(z|m',\lambda)\}}\sum_{m',\lambda}\sum_{x} p_x p(\lambda)p_e(m'|x,\lambda) p_d(z=x|m',\lambda).\nonumber\\
\eea
As $\sum_z p(z=x|m',\lambda)=1$ for all $m'$, the above definition can be presented as, 
 \bea\label{dcc11}
  \mathcal{D}_C            &=& \sum_{m',\lambda}\max_{x}\bigg\{ p_x p(\lambda)p_e(m'|x,\lambda)\bigg\}.
\eea
Using Bayes rule for conditional probability, we can write,  $p_e(m'|x,\lambda)=p_e(m',\lambda|x)/p(\lambda|x)$. As input $x$ is randomly chosen, $\lambda$ and $x$ are not correlated, i.e, $p_e(m',\lambda|x)=p(\lambda)p_e(m'|x,\lambda)$. Leveraging this into \eqref{dcc11}, it becomes,
\bea\label{dcc1}
 \mathcal{D}_C &=& \sum_{m',\lambda}\max_{x}\bigg\{ p_x p_e(m',\lambda|x)\bigg\} .
 \eea

For the uniform distribution $p_x=1/X$, the above simplifies to
\be \label{DC1}
\mathcal{D}_C =
\frac{1}{X}\sum_{m',\lambda}\max_{x}\bigg\{ p_e(m',\lambda|x)\bigg\}.
\ee 

As we are not imposing any restriction on the dimension of $m'$, the variable for shared randomness can be absorbed into the message $m'$. So, \eqref{DC1} can be written as, \\
\be \label{DC}
\mathcal{D}_C =
\frac{1}{X}\sum_{m}\max_{x}\bigg\{ p_e(m|x)\bigg\},
\ee 
where $(m',\lambda)=m.$ So classical distinguishability is the same as the case that is without shared randomness.
Hence, the conditional probability can be expressed as,
\bea
p(z|x,y) = \sum_m\sum_{x,y} p_e(m|x) p_d(z|y,m),
\eea
and consequently, the success merit,
\bea \label{scd}
\mathcal{S}_C &=& \max_{\substack{\{p_e(m|x)\},\\ \{p_d(z|y,m)\} }} \sum_{x,y,z} \sum_m c(x,y,z)  p_e(m|x) p_d(z|y,m) \nonumber\\
\eea
As $\sum_z p_d(z|y,m)=1$, for all $y,m$, the above expression can be expressed with only the encoding strategy, which is described below.
\be\label{scd1}
\mathcal{S}_C = \max_{\{p_e(m|x)\} } \sum_{y,m} \max_z \left\{  \sum_x c(x,y,z)  p_e(m|x)   \right\}.
\ee 

On the other hand, if we restrict the distinguishability to that which depends only on the communicated message, it can be defined mathematically as
\bea\label{eve_d}
\mathcal{D}'_C &=& 
 \sum_{m'}\max_{x}\bigg\{ p_x p_e(m'|x)\bigg\}.
\eea
It is important to note that imposing a constraint on this distinguishability, \(\mathcal{D}'_C\), does not render the task non-trivial.
\begin{lemma}\label{1}
   In classical communication, even if the distinguishability associated with the classical message—without conditioning on shared randomness—is constrained to its minimum value, there still exists an encoding strategy through which Bob can recover complete information about the input $x$.
\end{lemma}
\begin{proof}
We provide an encoding strategy where $\mathcal{D}'_C$ is the minimum , i.e., $\mathcal{D'}_C=\max_x\{p_x\}$ but distinguishability with respect to Bob, i.e., $\mathcal{D}_C$ becomes $1$.
Consider an encoding strategy with  $m'\in [X],\ \lambda\in[X]$ and $p(\lambda)=1/X$. Alice executes the strategy such that $p_e(m'|x,\lambda)=\delta_{m',x\oplus_X(\lambda-1)}$, where $\oplus_X$ denotes modulo sum $X$. In this strategy, for all $x$,
\bea
p_e(m'|x) &=& \sum_\lambda p(\lambda)p(m'|x,\lambda)\nonumber\\
&=& \frac1X \sum_\lambda \delta_{m',x\oplus_X(\lambda-1)}=\frac1X .
\eea
Substituting this in \eqref{eve_d}, we get
\bea
\mathcal{D}'_C &=& \frac1X\sum_{m'}\max\{p_x\}=\max\{p_x\}.
\eea
However, the quantity
\be
\max_x\{p_e(m'|x,\lambda)\} = \max_x\{\delta_{m',x\oplus_X(\lambda-1)}\}
= 1, 
\ee
for all $\lambda,m'.$
Replacing the above into \eqref{dcc11}, we find
\bea
\mathcal{D}_C &=& \sum_{m',\lambda}p(\lambda) \max_x\{p_x\delta_{m',x\oplus_X(\lambda-1)}\}\nonumber\\
&=&\sum_{m',\lambda} p(\lambda)p_x=1.
\eea
Therefore, Bob can identify Alice's inputs perfectly, and classical communication reaches the highest success probability of the task.   
\end{proof}

\subsection{Quantum Communication (QC)}
In a quantum communication scenario, Alice transmits a quantum state 
$\rho_x$ acting on $\mathbb{C}^d$, and Bob produces an outcome $z$ by 
performing a quantum measurement described by a collection of 
positive semidefinite operators $\{N_{z|y}\}_{z,y}$, where 
$z \in [Z]$ and $y \in [Y]$. This leads to the conditional statistics
\begin{equation} \label{pD}
    p(z|x,y) = \operatorname{Tr} \!\left( \rho_x N_{z|y} \right).
\end{equation}
For the circuit diagram, check FIG. \ref{figa}.

The distinguishability of the sender’s ensemble of quantum states 
$\{\rho_x\}$ is defined as
\begin{equation} \label{DistinguishabilityQuantum}
\begin{aligned}
 \mathcal{D}_Q= & \mathcal{DS}[\{\rho_x\}_x,\{1/X\}_x]\\
    = & \underset{\{Q_x\}}{\max}\ \  
    \frac1X \sum_x \, \operatorname{Tr}\!\left(\rho_x Q_x\right),
\end{aligned}
\end{equation}
where $\{Q_x\}_x$ are POVM (Positive Operator Valued Measurement) elements of the optimum measurement to distinguish the states $\{\rho_x\}_x$ sampled from equal probability distribution.

While such quantum correlation has been studied in \cite{PRR_24}, in this work, our main focus is on the entanglement assisted quantum correlations.
\subsection{Entanglement Assisted Classical Communication (EACC)}
In this communication task, Alice and Bob share an arbitrary pure entangled state $\sigma_{AB}\in \mathcal{H}_A\otimes\mathcal{H}_B$. Alice receives an input $x\in[X]$ and encodes her input using her part of entangled state $\rho_{AB}$ and send the $d$-labelled classical message $m$. Bob receives the input $y$ and depending on $y$ and $m$, he executes a measurement on his part of the entangled state and gives the output $z$. 

 Alice can use a general $A$-outcome POVM  described by the positive operators $\{M_{a|x}\}_x$ for encoding any input $x$ and $a\in[A]$. After implementing the specific measurement, Alice will send some classical message $m$ to Bob. Alice's message is an arbitrary function $f(a,x)$ of the output $a$ and the input $x$. The schematic representation of the task is provided in FIG. \ref{figb}.

The information available to Bob is composed of the classical message of Alice and reduced state at Bob's side after Alice's measurement. The reduced state at Bob's side is,
 \be \label{rhoBax}
\rho^B_{a|x} = \tr_A\left[\frac{\left(M_{a|x}\otimes \I\right)\sigma_{AB}\left(M^\dagger_{a|x}\otimes \I\right)}{p(a|x,\rho)}  \right],
\ee 
where $p(a|x,\rho) = \tr\left( \sigma_{AB} \left( M_{a|x}^\dagger M_{a|x} \otimes \I \right) \right)$.
 Depending on Alice's classical message $m$, Bob can perform a general $Z$-outcome measurement $\{N_{z|y,f(a,x)}\}_{y,f(a,x)}$ for some input $y$ on the reduced state of his side and $z\in[Z]$. The resulting correlation in this scenario becomes, 
\be 
p(z|x,y)=\sum_{a} \tr\left[\left(M_{a|x}\otimes N_{z|y,f(a,x)}\right)\sigma_{AB}\right].
\ee

 As Alice knows the shared state, the communication constraint can be defined as following:
\begin{widetext}
\bea 
\mathcal{D}^\sigma_{EACC}&=&\sum_x p_x p(z=x|x) \nonumber\\
&=& \sum_{x,a} p_x p\left(a|x,\sigma\right)p\left(b=x|a, f(a,x)\right)\nonumber\\
&=&  \sum_{a} \max_{\{Q'_{b|a,f(a,x)}\}} \sum_{x} p_x p\left(a|x,\sigma\right) \mbox{Tr} \left(\sigma^{B}_{a|x}\otimes \ket{f(a,x)}\bra{f(a,x)} Q_{b=x|a,f(a,x)}\right)\nonumber\\
&=&\sum_a \mathcal{DS}\left[\left\{\sigma^{B}_{a|x}\otimes \ket{f(a,x)}\!\bra{f(a,x)} \right\}_x,\left\{p_x p\left(a|x,\sigma\right) \right\}_x\right].
\eea    
\end{widetext}
 If $f(a,x)$ is independent of $x$, then for every $a$, the constraint depends on the distinguishability of states 
  $\sigma^{B}_{a|x}\otimes \ket{f(a)}\!\bra{f(a)}$. 
  For every $a$, it becomes the distinguishability of $\sigma^{B}_{a|x}$ with an ancillary state $\ket{f(a)}\!\bra{f(a)}$. Henceforth, the above quantity reduces to,
\bea\label{DME_sig}
\mathcal{D}^\sigma_{EACC}&=& \sum_a \mathcal{DS}\left[\left\{\rho^{B}_{a|x} \right\}_x,\left\{p_x p\left(a|x,\sigma\right) \right\}_x\right].\nonumber\\
\eea

\begin{thm}\label{th1}
    In this scenario, if Alice's inputs are random (i.e., $p_x=\frac1N$) and she uses the maximally entangled state $\ket{\phi_d^+}=\frac{1}{\sqrt{d}}\sum_{i=1}^d\ket{ii}$ and sends the outcome as the classical message to Bob, the highest value of  $\mathcal{D}^{\phi_d^+}_{EACC}$ for rank one projective measurements is $d/N$, where $N$ is the number of inputs. 
\end{thm}
\begin{proof}
      As the outcomes of the measurements are random, Alice's classical message does not increase the value of distinguishability of Alice's measurements.
     We want to find the partial distinguishability of $N$ projective measurements defined as $\left\{\ket{\eta_{a|x}}\bra{\eta_{a|x}}\right\}_{a,x}$ when the probe is maximally entangled state, i.e., $\ket{\phi_d^+}=\frac{1}{\sqrt{d}}\sum_i\ket{ii}$.
    Partial distinguishability in this case follows as from \eqref{DME_sig}:
    
    \bea \label{DME_th1}
 \mathcal{D}^{\phi_d^+}_{EACC}
&=& \frac{1}{N} \sum_a \mathcal{DS}\left[\left\{\ket{\eta_{a|x}}^*\right\}_x,\left\{p(a|x,\ket{\phi_d^+}) \right\}_x\right]\nonumber\\
&=&  \frac{1}{N} \sum_a \mathcal{DS}\left[\left\{\ket{\eta_{a|x}}^*\right\}_x,\left\{\frac1d \right\}_x\right].
\eea    

It can be easily checked that $p(a|x,\ket{\phi_d^+})=\frac1d$, $\forall a,x$ as $\{\eta_{a|x}\}_{a,x}$ is projective measurement and maximally entangled state can be written in any basis. By a straight forward calculation, \eqref{DME_th1} can be written as,
 \bea \label{dme_th1_2}
 \mathcal{D}^{\phi_d^+}_{EACC}
&=&\frac1N\sum_a\left(\sum_x\frac1d\tr\left(\ket{\eta_{a|x}}^*\!\bra{\eta_{a|x}}^* Q_{x|a}\right) \right),\nonumber\\
\eea    
where $\{Q_{x|a}\}_x$ are the POVM elements of the optimum measurement to distinguish the set of states $\left\{\ket{\eta_{a|x}}^*\right\}_x$ for a fixed $a$.
The eq. \eqref{dme_th1_2} can be written as,
\bea
&& \mathcal{D}^{\phi_d^+}_{EACC}\left[\left\{\ket{\eta_{a|x}}\!\bra{\eta_{a|x}}\right\}_{a,x},\left\{\frac{1}{N}\right\}_x\right]  \nonumber\\
&\leqslant & \frac1N\sum_a\sum_x\frac1d\tr\left(\ket{\eta_{a|x}}^*\!\bra{\eta_{a|x}}^*\right) \tr(Q_{x|a})\nonumber\\
&= & \frac{1}{Nd}\sum_a\sum_x\tr(Q_{x|a})\nonumber\\
&=&\frac{d}{N}.
\eea
As $\{Q_{x|a}\}_x$ are valid POVM elements, $\sum_x Q_{x|a} =\I$ for a fixed $a$ and that follows $\sum_x\tr(Q_{x|a})=d$ and the sum over $a$ ultimately gives the value $d^2$.
\end{proof}

 Note that, this proof can be extended for the cases where $p(a|x,\phi_d^+)$ are same for all $x$. This is possible if $p(a|x,\phi_d^+)=k/N$, where $k$ is the rank of $\ket{\eta_{a|x}}^*$ for each $a$.

\emph{Distinguishability in adversarial scenario.---} One can think of an \emph{adversarial scenario} where a third party (referee) chooses the entangled state $\rho_{AB}$ for the communication task. The information of shared entangled state is hidden from the parties. As Alice does not know about the entangled state, she will pursue the best strategy to hide her inputs as well as to complete the task. The distinguishability takes the following form:
\begin{widetext}
\bea \label{DME1}
 \mathcal{D}_{EACC}&=&\sum_x p_x p(z=x|x) \nonumber\\
&=& \sum_{x,a} p_x p\left(a|x,\rho\right)p\left(b=x|a, f(a,x)\right)\nonumber\\
&=& \max_{\rho_{AB}} \sum_{a} \max_{\{Q_{b|a,f(a,x)}\}} \sum_{x} p_x p\left(a|x,\rho\right) \mbox{Tr} \left(\rho^{B}_{a|x}\otimes \ket{f(a,x)}\bra{f(a,x)} Q_{b=x|a,f(a,x)}\right)\nonumber\\
&=& \max_{\rho_{AB}} \sum_a \mathcal{DS}\left[\left\{\rho^{B}_{a|x}\otimes \ket{f(a,x)}\bra{f(a,x)} \right\}_x,\left\{p_x p\left(a|x,\rho\right) \right\}_x\right].
\eea    
\end{widetext}

 Now, if $f(a,x)$ is independent of $x$, then the message does not contain any information of $x$ and consequently, \eqref{DME1} simplifies to
\bea
\mathcal{D}_{EACC}&=& \max_{\rho_{AB}} \sum_a \mathcal{DS}\left[\left\{\rho^{B}_{a|x} \right\}_x,\left\{p_x p\left(a|x,\rho\right) \right\}_x\right].\nonumber\\
\eea
It is interesting to note that this is nothing but the distinguishability of measurements $\{M_{a|x}\}_{a,x}$\cite{Manna_111,Datta_2021,ziman}.

\subsection{Entanglement Assisted Quantum Communication (EAQC)}
In this communication task, Alice and Bob share an arbitrary entangled state $\sigma_{AB}\in \mathcal{H}_A\otimes\mathcal{H}_B$. Alice receives an input $x\in[X]$.
 Alice can use a local CPTP map $\{E_{x}\}_x$ for encoding any input $x$. for every $x$, Kraus representation of a quantum channel  is defined by a set of operators $\{E_{a|x}\}_x$,  where $\sum_a E^\dagger_{a|x} E_{a|x}=\I$. After implementing the specific channel, Alice send her subsystem to Bob.
  Now the receiver holds both shares, one which Alice sent and another his share. He performs a measurement $\{N_{z|y}\}_{y}$ depending on his input $y$. The
correlations are given by the Born rule,
\bea  
 p(z|x,y) &=& \tr \left[\left(\sum_a(E_{a|x}^\dagger\otimes\I)\sigma_{AB}(E_{a|x}\otimes\I)\right)N_{z|y}\right].\nonumber\\
\eea

Note that if $\sigma^A_x$, the reduced density matrix of Alice, acts in $\mathbbm{C}^d$, then the output state acts on $\mathbbm{C}^{d'}$ that may not be in the same dimension. Thus, the Kraus operators acts as $\mathbbm{C}^d\rightarrow \mathbbm{C}^{d'},$ and the measurements by Bob $N_{z|y}$ acts on $\mathbbm{C}^{d'd}$. The schematic diagram is presented in FIG. \ref{figc}.



The communication constraint of this scenario is described as,
\bea
&&\mathcal{D}^\sigma_{EAQC}\nonumber\\
&=& \max_{\{T'_{b}\}_b}\sum_x p_x \tr \left[\left(\sum_a(E_{a|x}^\dagger\otimes\I)\sigma_{AB}(E_{a|x}\otimes\I)\right)Q_{b=x}\right]\nonumber\\
&=& \mathcal{DS}\left[\left\{\sum_a(E_{a|x}^\dagger\otimes\I)\sigma_{AB}(E_{a|x}\otimes\I)\right\}_x,\left\{p_x\right\}\right].\nonumber\\
\eea

Under the alternative definition of distinguishability, the communication task trivializes in a manner similar to CC. In particular, one can establish an analogue of Lemma \ref{1} to certify this behavior, with the additional feature that the number of communicated messages required is strictly smaller than in the corresponding CC scenario. In the EAQC setting, the parties share a maximally entangled state. Alice encodes her input $x \in [X]$ via unitary operations of dimension $\lceil \sqrt{X} \rceil$, following the dense coding protocol~\cite{benett}. Under the alternative definition, where distinguishability depends only on the communicated message, these unitaries-evolved states correspond to $X$ maximally mixed states, yielding distinguishability $1/X$. In contrast, $\mathcal{D}^{\phi^+}_{\mathrm{EAQC}} = 1$, since the post-encoding states form a set of $X$ mutually orthogonal states. Hence, Bob can perfectly recover Alice's input and achieve the optimal success probability.

We present a general lemma about the distinguishability of unitary channels. This would be helpful in our results demonstrating the advantage in EAQC protocol.
\begin{lemma}\label{lem1}
    In this scenario, if Alice's inputs are random and she uses the $d$-dimensional unitary channels to encode her messages and the parties share entangled state of dimension $d\times d$, the highest value of $\mathcal{D}^\sigma_{EAQC}$ is $d^2/N$, where $N$ is the number of inputs of Alice.
\end{lemma}
\begin{proof}
    This lemma is the consequence of the following two facts:\\
    First, for distinguishability in an entanglement-assisted scenario, the sufficient initial entangled state is of $d\otimes d$ dimension for $d$-dimensional unitaries \cite{bsxv-q9x7}.\\
    Second, the highest value of distinguishability of $N$ quantum states that belong to  $\mathbbm{C}^d$ sampled form uniform distribution is $\mathcal{D}_Q \leqslant \frac{d}{N}$.
\end{proof}

\emph{Distinguishability in adversarial scenario.---}
The communication constraint in this scenario is defined as,
\bea
&&\mathcal{D}_{EAQC}\nonumber\\
&=&\sum_x p_x p(z=x|x) \nonumber\\
&=& \max_{\rho_{AB},\{T_{b}\}_b}\sum_x p_x \tr \left[\left(\sum_a(E_{a|x}^\dagger\otimes\I)\rho_{AB}(E_{a|x}\otimes\I)\right)Q_{b=x}\right]\nonumber\\
&=& \max_{\rho_{AB}}\mathcal{DS}\left[\left\{\sum_a(E_{a|x}^\dagger\otimes\I)\rho_{AB}(E_{a|x}\otimes\I)\right\}_x,\left\{p_x\right\}\right],\nonumber\\
\eea
where $\rho_{AB}$ is the shared state by a third party.
 It can be easily noticed that the above expression is nothing but the distinguishability of channels \cite{AYuKitaev_1997,10.5555/3240076,piani,acin2,manna2025,manna2026,npj,duan,watrous}.

\begin{figure}
    \centering
    \begin{subfigure}[b]{0.4\textwidth}
    \centering
   \begin{tikzpicture}[scale=1,
  decoration={brace,mirror,amplitude=6pt}]

\node[left] at (-1.7,0) {Alice};

\draw (-1.7,0) -- (-0.9,0);

\draw (-0.9,-0.4) rectangle (0.0,0.4);
\node at (-0.4,0) {$U_x$};

\draw[double, -{Latex[length=2mm]}] (-0.4,0.8) -- (-0.4,0.4);
\node at (-0.4,1) {$x$};

\draw (0.0,0) -- (0.5,0);

\draw (0.5,0) -- (1.4,-1.5);
\draw (1.4,-1.5) -- (2.2,-1.5);

\draw[dashed] (-2,-0.7) -- (4,-0.7);

\node[left] at (-1.7,-1.6) {Bob};

\draw (2.2,-1.9) rectangle (3.2,-1.1);
\draw (2.4,-1.6) arc[start angle=180,end angle=0,radius=0.35 cm];
\draw[->, thick] (2.6,-1.6) -- (2.9,-1.35);
\node at (2.7,-2.3) {$\{M_{z|y}\}$};

\draw[double, -{Latex[length=2mm]}] (2.7,-0.5) -- (2.7,-1.1);
\node at (2.7,-0.3) {$y$};

\draw[double, -{Latex[length=2mm]}] (3.2,-1.5) -- (3.8,-1.5);;
\node[right] at (3.8,-1.5) {$z$};

\end{tikzpicture}

    \caption{Quantum Communication (QC)}
    \label{figa}
    \end{subfigure}
    \hfill
    \begin{subfigure}[b]{0.4\textwidth}
    \centering
\begin{tikzpicture}[scale=1,
  decoration={brace,mirror,amplitude=6pt}]

\node[left] at (-2.2,0) {Alice};

\draw (-1.8,0) -- (-1.0,0);

\draw (-1.0,-0.4) rectangle (0,0.4);
\draw (-0.85,-0.1) arc[start angle=180,end angle=0,radius=0.35 cm];
\draw[->, thick] (-0.75,-0.2) -- (-0.4,0.1);
\node at (-0.5,-0.7) {$\{M_{a|x}\}$};

\draw[double,-{Latex[length=2mm]}] (-0.5,0.9) -- (-0.5,0.4);
\node at (-0.5,1.1) {$x$};

\draw[double](-0.0,0) -- (0.6,0);
\node at (0.3,0.25) {$a$};

\draw (0.6,-0.3) rectangle (1.6,0.3);
\node at (1.1,0) {$f(a,x)$};

\draw[double] (1.6,0) -- (2,0);
\node at (1.9,0.25) {$m$};

\draw[double] (2.0,0) -- (2.7,-1.6);
\draw[double] (2.7,-1.6) -- (3.1,-1.6);

\draw[dashed] (-2.3,-1) -- (4.5,-1);

\node[left] at (-2.2,-2) {Bob};

\draw (-1.9,-2) -- (3.1,-2);

\draw (3.1,-2.3) rectangle (4.1,-1.5);
\draw (3.25,-2) arc[start angle=180,end angle=0,radius=0.35 cm];
\draw[->, thick] (3.4,-2.1) -- (3.7,-1.8);
\node at (3.65,-2.7) {$\{N_{z|y}\}$};

\draw[double,-{Latex[length=2mm]}] (3.6,-0.9) -- (3.6,-1.5);
\node at (3.6,-0.8) {$y$};

\draw[double,-{Latex[length=2mm]}] (4.1,-1.9) -- (4.6,-1.9);
\node[right] at (4.62,-1.9) {$z$};

\draw[decorate,decoration={brace,mirror,amplitude=8pt}]
(-1.9,0.4) -- (-1.9,-2.2)
node[midway,xshift=-1 cm] {$\rho_{AB}$};

\end{tikzpicture}

    \caption{Entanglement assisted Classical Communication (EACC)}
    \label{figb}
    \end{subfigure}\\
    \hfill
    
     \begin{subfigure}[b]{0.4\textwidth}
     \centering
    \begin{tikzpicture}[scale=1,
  decoration={brace,mirror,amplitude=6pt}]

\node[left] at (-2.1,0) {Alice};

\draw (-1.8,0) -- (-1.0,0);

\draw (-1.0,-0.45) rectangle (0.0,0.45);
\node at (-0.5,0) {$E_x$};

\draw[double,-{Latex[length=2mm]}] (-0.5,0.9) -- (-0.5,0.45);
\node at (-0.5,1.1) {$x$};

\draw (0,0) -- (0.5,0);

\draw (0.5,0) -- (1.5,-1.6);
\draw (1.5,-1.6) -- (2.4,-1.6);

\draw[dashed] (-2.3,-0.95) -- (4.1,-0.95);

\node[left] at (-2.1,-2.0) {Bob};

\draw (-1.8,-2.0) -- (2.4,-2.0);

\draw (2.4,-2.3) rectangle (3.4,-1.5);

\draw (2.6,-2.05) arc[start angle=180,end angle=0,radius=0.35];
\draw[->,thick] (2.8,-2.1) -- (3.1,-1.9);

\node at (3,-2.7) {$\{N_{z|y}\}$};

\draw[double,-{Latex[length=2mm]}] (3,-1.0) -- (3,-1.55);
\node at (3,-0.8) {$y$};

\draw[double,-{Latex[length=2mm]}] (3.4,-1.95) -- (3.8,-1.95);
\node[right] at (3.8,-1.95) {$z$};

\draw[decorate,decoration={brace,mirror,amplitude=8pt}]
(-1.9,0.4) -- (-1.9,-2.2)
node[midway,xshift=-1.1cm] {$\rho_{AB}$};
\end{tikzpicture}    
\caption{Entanglement assisted Quantum Communication (EAQC)}
    \label{figc}
     \end{subfigure}
  \caption{Circuit diagrams illustrating three different communication protocols. In the circuits, single lines represent quantum systems, while double lines denote classical variables. The dashed line denotes spatial separation between the parties.}     
  \label{fig}
\end{figure}    

\subsection{Quantifying Advantage}
In this work, our objective is to measure the quantum advantage in communication  with respect to the distinguishability of $x$. For a given value of a success metric $\mathcal{S}$ for a certain task, if the minimum value of distinguishability $\mathcal{D}_C$ in CC surpasses the value of distinguishability $\mathcal{D}_{EACC}$ or $\mathcal{D}_{EAQC}$, the quantum advantage is established for EACC and EAQC respectively. One can consider the hierarchy between EACC and EAQC too by considering the distinguishability of the both cases such that the value of success metric is same. But we are not considering this kind of advantage criteria as we are interested in finding the advantage in EACC and EAQC with respect to CC. To quantify this advantage, we draw an analogy from the standard notion of communication complexity and consider the ratio of the distinguishabilities in two kinds of communication, 
$\mathcal{D}_C^\mathcal{S} / \mathcal{D}_{EACC}^\mathcal{S} (\mathcal{D}_C^\mathcal{S} / \mathcal{D}_{EAQC}^\mathcal{S})$, given that the success metric attains at least a certain value, $\mathcal{S}$, in both, CC and EACC (EAQC). If this ratio exceeds 1, it indicates a quantum advantage. 

The quantum advantage can also be quantified in the reverse direction by considering the ratio of the success metrics,
$\mathcal{S}_{EACC}^{\mathcal{D}} / \mathcal{S}_C^{\mathcal{D}}$ (or $\mathcal{S}_{EAQC}^{\mathcal{D}} / \mathcal{S}_C^{\mathcal{D}}$),
under the constraint that the distinguishability $\mathcal{D}$ is upper bounded by a fixed value, for both CC and EACC (EAQC).
In certain cases, deriving the general relation of success merit as a function of distinguishability is analytically intractable.
In such situations, we evaluate the success probability by fixing the distinguishability of the sender's input states numerically,
and this alternative quantification proves to be particularly useful.

 \section{EACC, EAQC and QC are equivalent}
In this section, we prove that all three protocols EACC, EAQC and QC are basically equivalent in terms of producing correlations in restricted distinguishability. That automatically facilitates the fact that advantage with respect to CC in any of the scenario implies the advantage in the other two. We start by showing the equivalence of two entanglement assisted scenario. Note, there is no limit on the dimension of the message.
\begin{lemma}\label{obs1}
    Suppose there exists an EACC protocol involving an entangled state $\sigma_{AB} \in \mathbbm{C}^d\otimes \mathbbm{C}^d$ that produces $\{p(z|x,y)\}$ with distinguishability $\mathcal{D}^\sigma_{EACC}$. Then there exists an EAQC protocol involving $\sigma_{AB}\otimes \ket{\phi_d^+}\!\bra{\phi_d^+}_{A'B'}$ that produces the same $\{p(z|x,y)\}$ with the same distinguishability, where $\ket{\phi_d^+}$ is the maximally entangled state.
\end{lemma}
\begin{proof}
Let us take $R$ to be the number of possible messages $\{f(a,x)\}$ in the EACC protocol.  Consider the EAQC protocol where Alice implements the same measurement $\{M_{a|x}\}_{a,x}$ on her part of $\sigma_{AB}$ that is used in EACC protocol. Additionally$\ket{\phi_d^+}\!\bra{\phi_d^+}_{A'B'}$ is used to send the message $f(a,x)$ using the dense-coding protocol \cite{benett}, such that $\ket{\phi_d^+}_{A'B'} \in \mathbbm{C}^{\lceil{\sqrt{R}}\rceil}\otimes \mathbbm{C}^{\lceil{\sqrt{R}}\rceil}$. There is no other quantum communication apart from the dense-coding implementation. The distinguishability in the EAQC protocol depends on the distinguishability of collapsed state on Bob's side as well as the state Alice sends for quantum communication, i.e., $\sum_a \left\{\sigma^{B}_{a|x} \otimes (U_{a}^\dagger\otimes\I)\ket{\phi_d^+}\bra{\phi_d^+}(U_{a}\otimes\I)\right\}_x,$ with probability $\left\{p(a|x,\sigma) \right\}_x$, where $\{U_a\}_a$ are the unitary operators used in dense coding protocol. Mathematically, we can write,
\bea
&&\mathcal{D}^{\sigma\otimes\ket{\phi_d^+}\bra{\phi_d^+}}_{EAQC}\nonumber\\
&=&\sum_a \mathcal{DS}\left[\left\{\sigma^{B}_{a|x}\otimes (U_{a}^\dagger\otimes\I)\ket{\phi_d^+}\bra{\phi_d^+}(U_{a}\otimes\I) \right\}_x,\right.\nonumber\\
&&\left.\left\{p_x p\left(a|x,\sigma\right) \right\}_x\right]\nonumber\\
&=& \sum_a \mathcal{DS}\left[\left\{\sigma^{B}_{a|x}\right\}_x,
\left\{p_x p\left(a|x,\sigma\right) \right\}_x\right]\nonumber\\
&=& \mathcal{D}^{\sigma}_{EACC}
\eea
Last line comes from the fact that for every $a$, it becomes the distinguishability of $\sigma^{B}_{a|x}$ with an ancillary state $(U_{a}^\dagger\otimes\I)\ket{\phi_d^+}\bra{\phi_d^+}(U_{a}\otimes\I)$. So distinguishability
remains the same as in the EACC protocol.    
\end{proof}

\begin{lemma}\label{obs2}
    Suppose there exists an EAQC protocol involving an entangled state $\sigma_{AB} \in \mathbbm{C}^d\otimes \mathbbm{C}^d$ and channels $\{E_x\}_x$ ($\mathbbm{C}^{d}\rightarrow \mathbbm{C}^{d'}$) that produces $\{p(z|x,y)\}$ with distinguishability $\mathcal{D}^\sigma_{EAQC}$. Then there exists an EACC protocol involving $ \sigma_{AB}\otimes\ket{\phi_{d'}^+}\!\bra{\phi_{d'}^+}_{A'B'}$ where $\ket{\phi_{d'}^+}_{A'B'}\in \mathbbm{C}^{d'}\otimes \mathbbm{C}^{d'}$ that produces the same $\{p(z|x,y)\}$ with the same distinguishability.
\end{lemma}
\begin{proof}
In EACC protocol, Alice first applies $E_x$ on her part of $\sigma_{AB}$ depending on $x$. Then she teleports $\sigma_A$ to Bob using the teleportation scheme \cite{tele}. The teleportation will need $d'^2$ dimensional classical message. Note that, we do not constrain the dimension of messages in the scenario. In teleportation protocol, all the outcomes are totally random. So it does not change the sampling probability of channels $\{E_x\}$. The distinguishability only comes from the states $\left\{\left(\sum_a(E_{a|x}^\dagger\otimes\I)\sigma_{AB}(E_{a|x}\otimes\I)\right)\right\},\{p_x\}$. Thus, the distinguishability in EACC, which is $\mathcal{D}^{\sigma\otimes\ket{\phi_{d'}^+}\!\bra{\phi_{d'}^+}}_{EACC}$
reduces to $\mathcal{D}^{\sigma}_{EAQC}$.
\end{proof}

Next, we will prove the equivalence between two types of quantum communication scenarios. 
\begin{lemma}\label{l5}
     Suppose there exists an EAQC protocol involving an entangled state $\sigma_{AB} \in \mathbbm{C}^d\otimes \mathbbm{C}^d$ and channels $\{E_x\}_x$ ($\mathbbm{C}^{d}\rightarrow \mathbbm{C}^{d'}$) that produces $\{p(z|x,y)\}$ with distinguishability $\mathcal{D}^\sigma_{EAQC}$.  Then there exists an QC protocol involving the transmitted states $\{(E_x\otimes\I)\sigma_{AB}(E_x^\dagger\otimes\I)\}_x$ which produces the same $\{p(z|x,y)\}$ with the same distinguishability.
\end{lemma}
\begin{proof}
    Suppose Alice implements channels $\{E_x\}_x$ for her inputs $x$ in EAQC. After applying these channels, she sends the resulting quantum state to Bob. Consequently, Bob receives the ensemble of states ${(E_x\otimes\I)\sigma_{AB}(E_x^\dagger\otimes\I)}x$. In this scenario, $\mathcal{D}^\sigma_{EAQC}$ denotes the distinguishability of the states ${(E_x\otimes\I)\sigma_{AB}(E_x^\dagger\otimes\I)}_x$.
In the QC protocol, Alice can encode her inputs into the same set of states ${\rho_x=(E_x\otimes\I)\sigma_{AB}(E_x^\dagger\otimes\I)}_x$, thereby reproducing the same correlations with identical distinguishability.
\end{proof}
\begin{lemma}\label{l6}
    Suppose there exists a QC protocol that produces $\{p(z|x,y)\}$ with distinguishability $\mathcal{D}_Q$.  Then there exists an EAQC protocol involving a shared product state $\sigma_{AB}=\sigma_A\otimes\sigma_B \in \mathbb{C}^d\otimes \mathbbm{C}^{d'}$  and channels $\{E_x\}_x$ ($\mathbbm{C}^{d}\rightarrow \mathbbm{C}^{d'}$) that produces the same $\{p(z|x,y)\}$ with the same distinguishability $\mathcal{D}^\sigma_{EAQC}$.
\end{lemma}
\begin{proof}
    In QC protocol, Alice sends quantum states $\{\rho_x\}_x$. In EAQC protocol, Alice can choose the channels $\{E_x\}_x$ such that $E_x\sigma_A E_x^\dagger=\rho_x$. In this protocol, Alice and Bob share a product state $\sigma_A\otimes\sigma_B$. The distinguishability depends on the distinguishability of $\{E_x\sigma_A E_x^\dagger\}_x$ supplemented with an ancillary state $\sigma_B$, which is distinguishability of $\{\rho_x\}_x$ and that is the distinguishability in QC.
\end{proof}
From last two lemmas, we are certain that all the strategies of EAQC can be executed in QC and vice versa. Earlier we prove that all the strategies of EAQC can be carried out in EACC and vice versa. In a summary, we can conclude that EACC,EAQC and QC are equivalent.

\section{No input on Bob}
In this scenario, Frenkel and Weiner~\cite{fw} showed that quantum communication provides no advantage over classical strategies in dimension-constrained communication tasks. Here, we address the same question under distinguishability constraint and find that the answer remains negative in this case as well.
\begin{thm}
    If Bob receives no input, then the set of conditional probabilities $\{p(z|x)\}$ achievable using QC coincides exactly with the set achievable using CC. Therefore, in such scenarios, quantum communication provides no advantage over classical communication. Consequently, EAQC and EACC offer no advantage over CC.
\end{thm}
\begin{proof}
We now prove the first statement. The next statement follows from the equivalence of three protocols QC, EACC and EAQC.

Consider a task in which Alice receives an input \(x\in [X]\), while Bob performs a single measurement and produces an outcome \(z\in [Z]\). Let the quantum strategy be specified by the encoding quantum states \(\{\rho_x\}_{x}\) and the POVM \(\{N_z\}_{z}\). The resulting conditional probability distribution is
\begin{equation}
    p_Q(z|x)=\tr(\rho_x N_z).
\end{equation}
Let the distinguishability of the states \(\{\rho_x\}_x\), as defined in Eq.~\eqref{DistinguishabilityQuantum}, be denoted by \(\mathcal{D}_Q\). We now construct a classical strategy that reproduces the same conditional probabilities \(p_Q(z|x)\) for all \(x,z\), while ensuring that the corresponding classical distinguishability does not exceed \(\mathcal{D}_Q\).

In the classical protocol, the communicated message is denoted by \(z\in[Z]\). The encoding strategy is defined as
\begin{equation}
    p_e(z|x)=\tr(\rho_x N_z),
\end{equation}
while the decoding strategy is chosen to be the identity decoder,
\begin{equation}
    p_d(z|z')=\delta_{z,z'}.
\end{equation}
Notice that $p_e(z|x)$ is a valid probability distribution as $p_e(z|x)\geq 0$ and $\sum_z p_e(z|x)=\sum_z \tr(\rho_x N_z)= \tr(\rho_x)=1$. Therefore, the conditional probabilities generated by the classical strategy are
\begin{equation}
p_C(z|x) = \sum_{z'} p_e(z'|x)\delta_{z,z'}
= p_Q(z|x), \ \forall x,z.
\end{equation}
We now compare the distinguishability of the two strategies. Let \(\{\tilde p_d(x|z)\}_{x,z}\) denote the optimal decoding strategy for distinguishing the inputs from the classical messages. The classical distinguishability, defined in Eq.~\eqref{DC}, is given by
 \bea \label{DC-ni}
\mathcal{D}_C &=&
\frac{1}{X}\sum_{z}\max_{x} \bigg\{ \tr(\rho_x N_z)\bigg\}\nonumber\\
&=&
\frac{1}{X}\sum_{z,x}\tr(\rho_x N_z)\tilde{p}_d(x|z) .
 \eea
Define the operators
\be 
Q_x=\sum_z \tilde{p}_d(x|z) \ N_z .
\ee 
One can check $\{Q_x\}_x$ is a POVM. We have $Q_x\geq 0$ since $N_z\geq 0$ and $\tilde{p}_d(x|z)\geq 0$. Moreover, $\sum_x Q_x=\sum_x\sum_z N_z \tilde{p}_d(x|z)=\sum_z N_z \sum_x \tilde{p}_d(x|z)=\sum_z N_z=\I$. Substituting \(Q_x\) in \eqref{DC-ni}, we obtain
 \be\label{dss}
\mathcal{D}_C =
\frac{1}{X}\sum_{x} \tr(\rho_x Q_x) \leq \mathcal{D}_Q.
 \ee
That means all the probabilities of quantum communication can be obtained from classical communication, and \eqref{dss} ensures that classical distinguishability is less than or equal to quantum distinguishability.

%
%

For the reverse scenario, every classical ensemble can
be embedded into quantum theory as a commuting ensemble. Given a classical message $m$, define
\be
\rho_x=\sum_m p_e(m|x)\ket{m}\!\bra{m};
\ee
and for every decoding function $p_d(z|m)$, define a POVM,
\be
N_z=\sum_m p_d(z|m)\ket{m}\!\bra{m}.
\ee
Then
\bea
\tr(\rho_x N_z)=\sum_m p_e(m|x)p_d(z|m).
\eea
For quantum distinguishability, the optimum measurement element $M_x$ can be similarly modeled as
\be
Q_x=\sum_m p_d(x|m)\ket{m}\bra{m}.
\ee
Therefore,
\bea
\mathcal{D}_Q&=&\underset{\{Q_x\}}{\max}\ \  
    \frac1X \sum_x \, \operatorname{Tr}\!\left(\rho_x Q_x\right)\nonumber\\
    &=& \frac1X \sum_m p_e(m|x)p_d(x|m)\nonumber\\
    &=& \frac1X \sum_m \max_x\left\{p_e(m|x)\right\}\nonumber\\
    &\leq& \mathcal{D}_C
\eea
Thus every feasible classical strategy is also a feasible
quantum strategy with identical statistics.
This completes the proof.
\end{proof}

\subsection*{Advantage in dimension restricted scenario}
Now we know that there does not exist any advantage in EACC and EAQC with respect to CC in this scenario. Therefore, we define a new class of protocols, which limits the dimension of the communicated system as well as the distinguishability of the sender's inputs. We allow the sender to transmit one bit (or one qubit) message to the receiver. 

Under this restriction, no advantage can arise when Alice has only two possible inputs or when Bob’s measurement has only two possible outputs. This is because, in such scenarios, it is sufficient to consider a single bit of classical communication, irrespective of the distinguishability constraint imposed on the communication \cite{Tavakoli2022}. We next consider the scenario in which Alice has three possible inputs, \( x \in \{1,2,3\} \), and Bob has three possible outputs, \( z \in \{1,2,3\} \). By deriving all the facet inequalities of the corresponding classical polytope, we find that the non-trivial facet inequalities are
\bea
  &(i)& \frac13\left(p(1|2)+p(2|3)+p(3|1)\right)\leq \mathcal{D};\nonumber\\
  &(ii)& \frac23 -\frac13\left(p(2|1)+p(1|2)+p(3|3)\right)\leq \mathcal{D};\nonumber\\
\eea
The first inequality correspond precisely to the optimal success probabilities for discriminating all three of Alice’s inputs. Since the performance of the task is bounded by the distinguishability of Alice’s inputs, these success measures cannot exhibit any advantage in EACC. For the second inequality, $\mathcal{S}_{EACC}=\frac23-\mathcal{D}_{EACC}$. To obtain an advantage, one would require $\frac23-\mathcal{D}_{EACC}>\mathcal{D}_{EACC}$, which implies $\mathcal{D}_{EACC}<1/3$, which is not possible.

It turns out that the next scenario, where Alice receives three possible inputs, \( x \in \{1,2,3\} \), while Bob performs a fixed measurement with four possible outcomes, \( z \in \{1,2,3,4\} \), constitutes the simplest scenario exhibiting an advantage in EACC. The resulting correlations are described by the conditional probability distribution \( p(z|x) \). We obtain all the facet inequalities in our scenario, and find that only one of them admits a quantum violation. This inequality is given below:
 \bea\label{f1}
\frac12\sum_{i=1}^3 p(i|i)+\frac14\sum_{i=1}^3 p(4|i) \leq \frac{3\mathcal{D}+2}{4} .
 \eea
 The same success merit becomes same with one at \cite{Vieira_2023} if we put distinguishability to be $2/3$. Now we describe a strategy in EACC. Alice and Bob share qubit-qubit Bell state $\ket{\phi^+_2}$. When Alice got inputs $x=1$ and $x=2$, she executes two projective measurements in the basis $\left\{\ket{\psi},\ket{\psi^\perp}\right\}$ and  $\left\{\ket{\phi},\ket{\phi^\perp}\right\}$ respectively and communicates the outcome $m=a\in\{1,2\}$, where $\ket{\psi}=\frac{\sqrt{3}}{2}\ket{0}+\frac12\ket{1}$ and $\ket{\phi}=-\frac{\sqrt{3}}{2}\ket{0}+\frac12\ket{1}$. For $x=3$, she sends $m=2$ deterministically. If Bob got $m=1$, he measures in $\sigma_x$ basis and he measures in $\sigma_z$ basis for $m=2$. For $\sigma_x$ measurement, Bob uses the following post-processing: if he gets outcome corresponding to $\ket{+}\!\bra{+}$, he outputs $1$ and for the other outcome, he produces $2$. Similarly, for $\sigma_z$ measurement, he produces output $3$ for the outcome corresponding to $\ket{0}\!\bra{0}$, otherwise his output is $4$. Using this strategy, $\mathcal{S}_{EACC}=\frac{13+2\sqrt{3}}{16}\approx 1.029$. To achieve this success merit, we need $\mathcal{D}_C\geq \frac{5+2\sqrt{3}}{12}$ (from \eqref{f1}).
 
Distinguishability of Alice's inputs in this scenario,
 \bea
\mathcal{D}^{\phi_2^+}_{EACC}&=& \frac13\left(\mathcal{DS}\left[\{\ket{\psi},\ket{\phi}\},\{1/2,1/2\}\right]\right.\nonumber\\
&+&\left.\mathcal{DS}\left[\left\{\ket{\psi^\perp},\ket{\phi^\perp},\frac{\I}{2}\right\},\{1/2,1/2,1\}\right]\right)\nonumber\\
&=& \frac{6+\sqrt{3}}{12}.
 \eea
Therefore advantage in EACC over CC would be $\mathcal{D}^{\mathcal{S}}_C/\mathcal{D}^{\mathcal{S}}_{EACC}\geq \frac{5+2\sqrt{3}}{6+\sqrt{3}}=1.095$, when $\mathcal{S}= \frac{13+2\sqrt{3}}{16}$.

Now we want to find the range of success merit for which the advantage in EACC exists. As generalization of EACC strategies is difficult, we stick to the same strategy but change the measurement. Needless to say, it is enough to change one measurement; in this case, we take a general form of $\ket{\Tilde{\phi}}=-\cos\frac{\theta}{2}\ket{0}+\sin\frac{\theta}{2}\ket{1}$. So, the success merit $\mathcal{\Tilde{S}}_{EACC}=\frac{15}{32}+\frac14\cos^2\frac{\pi}{12}+\frac18\cos^2\frac{\theta}{2}+\frac14\cos^2(\frac{\pi}{4}-\frac{\theta}{2})$. In a similar manner, advantage in EACC is 
 \bea
\frac{\mathcal{\Tilde{D}}^{\mathcal{\Tilde{S}}}_C}{\mathcal{\Tilde{D}}^{\mathcal{\Tilde{S}}}_{EACC}} \geq \frac{\frac13\cos^2\frac{\pi}{12}
        + \frac16\cos^2\frac{\theta}{2}
        + \frac13\cos^2(\frac{\pi}{4} - \frac{\theta}{2})
        - \frac{1}{24}}{\frac12+\frac16\sqrt{1-\cos^2(\frac{\pi}{6}+\frac{\theta}{2})}},\nonumber\\
 \eea
 when $\mathcal{\Tilde{S}}_{EACC}=\frac{15}{32}+\frac14\cos^2\frac{\pi}{12}+\frac18\cos^2\frac{\theta}{2}+\frac14\cos^2(\frac{\pi}{4}-\frac{\theta}{2}).$
 

Next, we plot the advantage, i.e., the ratio of distinguishabilities with respect to success probability $\mathcal{\Tilde{S}}$ at FIG. \ref{fig2}.
\begin{figure}[h!]
    \centering

    \begin{subfigure}{0.9\linewidth}
        \centering
        \includegraphics[width=\linewidth]{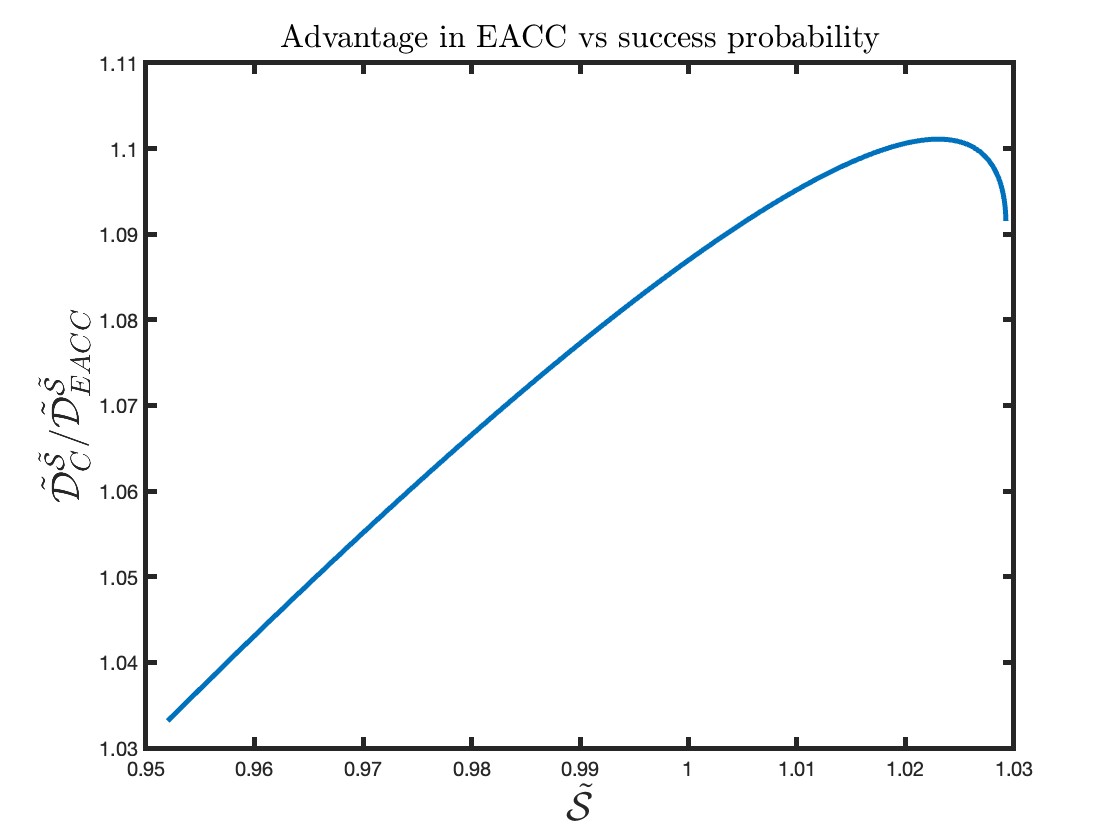}
        \caption{
        The plot illustrates the range of advantage of EACC with respect to the success probability. 
        }
        \label{fig2}
    \end{subfigure}

    \vspace{0.5cm}

    \begin{subfigure}{0.9\linewidth}
        \centering
        \includegraphics[width=\linewidth]{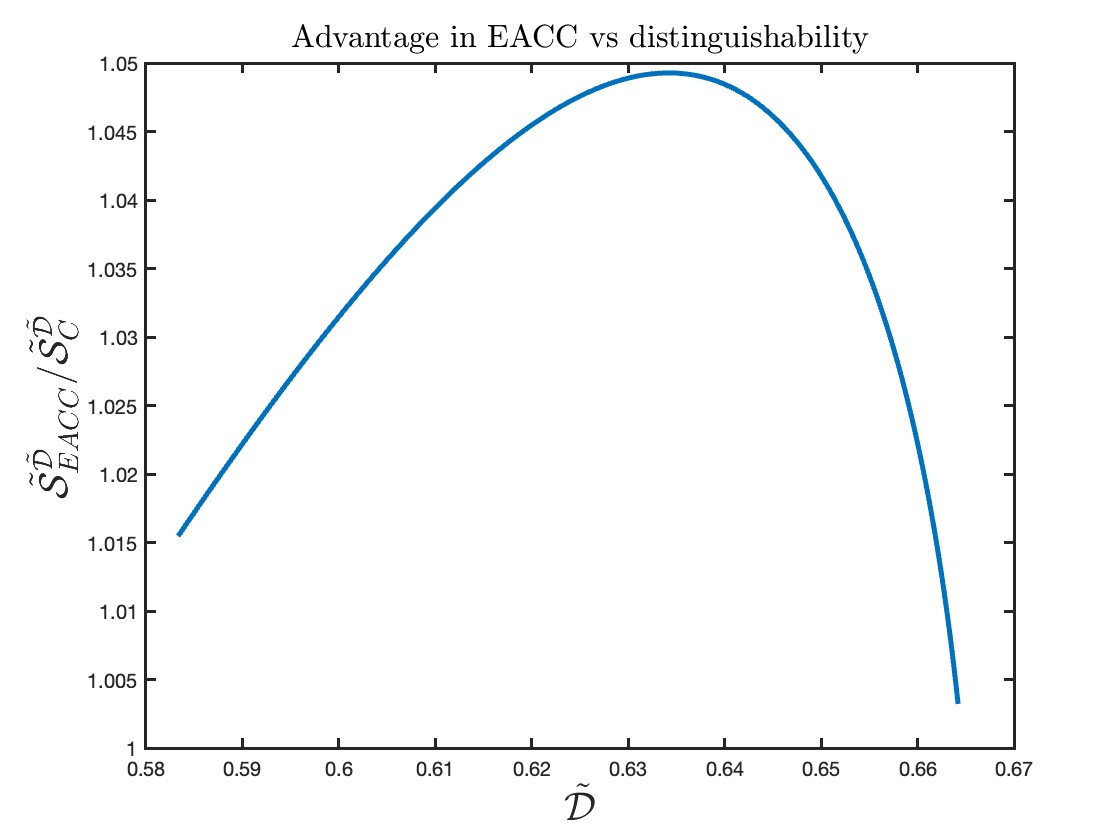}
        \caption{
        The plot illustrates the range of advantage of EACC with respect to distinguishability. 
        }
        \label{fig3}
    \end{subfigure}

    \caption{We compare the range of advantage of EACC with respect to success probability and distinguishability depending on the different measure of advantage. The parties share a maximally entangled state. We consider the same strategy with a different second measurement defined by the basis $\{\ket{\tilde{\phi}},\ket{\tilde{\phi}^{\perp}}\}$, parametrized by $\theta$. }
    \label{fig:eacc_advantage}
\end{figure}
We can define the advantage in other way around. In this case, we fix the distinguishability $\mathcal{\Tilde{D}}$ in both the communications, and find the ratio $\mathcal{\Tilde{S}}_{EACC}^\mathcal{\Tilde{D}}/\mathcal{\Tilde{S}}^\mathcal{\Tilde{D}}_C$. The next plot at FIG \ref{fig3} gives the range of distinguishability where the advantage appears in EACC.

Ref. \cite{svegborn} shows that if the parties share the state $\ket{\zeta}=\sqrt{\frac23}\ket{00}+\sqrt{\frac13}\ket{11}$, the success merit of \eqref{f1} reaches the value of $\frac{9+2\sqrt{3}}{12}$, which is greater than the success merit we got with maximally entangled state. In this case,
\bea
&&\mathcal{D}^{\zeta}_{EACC}\nonumber\\
&=& \frac13\left(\mathcal{DS}\left[\{\ket{\overline{\psi}},\ket{\overline{\phi}}\},\{\frac12+\frac{1}{6\sqrt{3}},\frac12+\frac{1}{6\sqrt{3}}\}\right]\right.\nonumber\\
&+&\left.\mathcal{DS}\left[\left\{\ket{\overline{\psi}^\perp},\ket{\overline{\phi}^\perp},\frac{\I}{2}\right\},\{\frac12-\frac{1}{6\sqrt{3}},\frac12-\frac{1}{6\sqrt{3}},1\}\right]\right)\nonumber\\
&=& 0.6604,
 \eea
where $\ket{\overline{\psi}}=\frac{\sqrt{2}\alpha\ket{0}+\beta\ket{1}}{\sqrt{2\alpha^2+\beta^2}}$ and $\ket{\overline{\phi}}=\frac{\sqrt{2}\alpha\ket{0}-\beta\ket{1}}{\sqrt{2\alpha^2+\beta^2}}$, where $\alpha=\sqrt{\frac{1+\frac{1}{\sqrt{3}}}{2}}, \beta=\sqrt{\frac{1-\frac{1}{\sqrt{3}}}{2}}$.

Advantage in this case $\mathcal{D}^{\mathcal{S}}_C/\mathcal{D}^{\mathcal{S}}_{EACC}\geq \frac{0.7182}{0.6604}=1.0875$, when $\mathcal{S}= \frac{9+2\sqrt{3}}{12}$. In this case, the advantage is smaller than that obtained with a maximally entangled state.
 

  In Ref.~\cite{fw}, it has been shown that quantum communication (QC) offers no advantage over classical communication (CC) when Bob receives no input and the dimension of the communicated system is the same in both cases. In this scenario, when the message dimension is $2$, the maximum achievable distinguishability in quantum communication is $2/3$. Substituting this value into Eq.~\eqref{f1}, we obtain the maximal value of the correlation as $1$, which is identical for both CC and QC. We already found that success probability of EACC exceeds the value of $1$ with distinguishability $\frac{6+\sqrt{3}}{12}$, which is smaller that $2/3$. Therefore, the performance of QC cannot exceed that of $\mathcal{S}_{EACC}$, leading to the conclusion that EACC is strictly superior to QC in this setting.

By Lemma~\ref{obs1}, a similar advantage can be established for EACC using the distinguishability $\mathcal{D}^{\phi_2^+}_{EACC}$. Furthermore, it can be observed that achieving the same advantage in EAQC requires not more than $2$-dimensional quantum message.

\section{Advantageous EACC with one bit message}
In this section, we find entanglement assisted protocols which are advantageous with respect to classical communication. We allow infinite amount of classical messages in CC but we find advantage in EACC by transmitting one bit classical message.
First, we exhibit a result which connects QC with EACC when transmitted message is a bit.
\begin{thm}\label{th2}
    Any QC protocol where Alice sends qubit pure states with Bob executes projective measurements can be implemented in EACC with same distinguishability, i.e., $\mathcal{D}_Q=\mathcal{D}^{\phi_2^+}_{EACC}$.
\end{thm}
\begin{proof}
   In a binary–output quantum communication task, Alice receives an input 
$x \in [X]$. For each input $x$, she sends a pure quantum state 
$\ket{\zeta_x}$ to Bob. Bob receives an input 
$y \in [Y]$ and performs a projective measurement in the basis 
$\{\ket{\kappa_y},\ket{\kappa^\perp_y}\}$ depending on $y$.  
The resulting conditional statistics are
\bea
    p(z|x,y) &=& \operatorname{Tr}\!\left(\ket{\zeta_x}\bra{\zeta_x} \ket{\kappa_y}\bra{\kappa_y}\right)\nonumber\\
    &=&|\langle \zeta_x|\kappa_y\rangle|^2.
\eea

Suppose now Alice and Bob share a maximally entangled two–qubit state 
$\lvert\phi_2^+\rangle_{AB}$.  
For each input $x$, Alice performs the binary projective measurement 
$\{\ket{\zeta_x}^*\!\bra{\zeta_x}^*,\;\ket{\zeta_x^\perp}^*\!\bra{\zeta_x^\perp}^*\}$,
and sends the classical outcome to Bob.

\begin{enumerate}
    \item If Alice obtains the outcome corresponding to $\ket{\zeta_x}^*$, then Bob’s post–measurement reduced state is $\ket{\zeta_x}$.  
    Upon receiving this outcome, he performs the measurement in the basis 
$\{|\kappa_y\rangle,|\kappa_y^\perp\rangle\}$.

    \item If Alice obtains the outcome corresponding to $\ket{\zeta_x^\perp}^*$, then Bob’s reduced state is  $\ket{\zeta_x^\perp}$.  
   Bob performs the same measurement but reverses the outcome, i.e., if Bob gets outcome $'0'$, he gives output $'1'$ and for outcome $'1'$, he generates output $'0'$.
\end{enumerate}
In this strategy, Bob's statistics will be,
\bea
&&p_{\mathrm{EACC}}(z|x,y)\nonumber\\
&=&\left[p(a=0|x)p(z|x,y,a=0)\right.\nonumber\\
&&\left.+p(a=1|x)p(\overline{z}|x,y,a=1)\right],
\eea
where $'a'$ denotes the outcome of measurement corresponding to input $x$.
As measurements are projective, we can write $p(a=0|x)=p(a=1|x)=1/2$. The above equation reads as,
\bea\label{pzxy2}
  &&  p_{\mathrm{EACC}}(z|x,y)\nonumber\\
    &=& \frac12\left[|\langle \zeta_x|\kappa_y\rangle|^2+|\langle \zeta^\perp_x|\kappa^\perp_y\rangle|^2\right]\nonumber\\
    &=& |\langle \zeta_x|\kappa_y\rangle|^2.
\eea
  So \eqref{pzxy2} gives the same statistics as QC if Bob reverses the outcome. 
Thus the success metric in the entanglement assisted classical communication protocol 
is the same as in the original quantum communication task.

To obtain the same advantage as in the quantum communication scenario, we must 
show that the distinguishability of the inputs in  EACC is equal 
to the distinguishability achieved through direct quantum communication. Alice only sends her outcome as the classical message, so distinguishability depends only on the measurements of Alice as the outcomes of Alice are random.
In this case, the distinguishability takes the form of Eq.~\eqref{DME_th1} with $d = 2$.  

By a straightforward calculation, we find
\bea
&& \mathcal{D}^{\phi_2^+}_{EACC}
   \nonumber\\
&=& \frac{1}{N} \left(
        \mathcal{DS}\!\left[
             \left\{\ket{\zeta_x}\right\}_x,
             \left\{\tfrac{1}{2}\right\}_x
        \right]
        +
        \mathcal{DS}\!\left[
             \left\{\ket{\zeta_x^\perp}\right\}_x,
             \left\{\tfrac{1}{2}\right\}_x
        \right]
    \right) \nonumber\\
&=& \frac{1}{2} \left(
        \mathcal{DS}\!\left[
             \left\{\ket{\zeta_x}\right\}_x,
             \left\{\tfrac{1}{N}\right\}_x
        \right]
        +
        \mathcal{DS}\!\left[
             \left\{\ket{\zeta^\perp_x}\right\}_x,
             \left\{\tfrac{1}{N}\right\}_x
        \right]
    \right) \nonumber\\
&=& \frac{1}{2}\left(\mathcal{D}_Q + \mathcal{D}_Q\right) \nonumber\\
&=& \mathcal{D}_Q ,
\eea
which coincides with the distinguishability of the inputs in QC. We exploit the fact that the distinguishability of a collection of qubit states sampled 
from a given probability distribution is equal to the distinguishability of the 
corresponding orthogonal states sampled from the same distribution.

Thus, the strategy described above achieves the same quantum advantage in the EACC as in QC. 
\end{proof}
Needless to say, this QC protocol can be implemented in EAQC with same distinguishability.

In the subsequent subsections, we will provide some specific examples of communication tasks and the advantage in EACC and EAQC over CC in those particular tasks. 
\subsection{Random Access Codes}
In this task, Alice is given a string of $n$ dits, denoted by $x = x_1 x_2 \cdots x_n$, chosen uniformly at random from the set of all 
$d$-ary strings of length $n$. Each component $x_y$ belongs to the alphabet $[d]$ for every $y \in [n]$. 
After communication, Bob's goal is to guess the $y$-th dit of the string, where $y$ is chosen uniformly at random from the set $[n]$. As the task depends on the parameters $n$ and $d$, we refer to it as the 
$(n,d)$-RAC.
 The success metric is determined by the average success probability, defined through 
\be \label{Snd}
\mathcal{S}(n,d) = \frac{1}{nd^n} \sum_{x,y} p(z=x_y|x,y) .
\ee
The ref. \cite{PRR_24} provides that, for  $(n,d)$-RAC, the following holds in CC,
     \be\label{scnd} 
n \mathcal{S}_{C}(n,d) +1-n \leqslant \mathcal{D}_C.
     \ee 
Now we list the advantage in EACC protocols in $(2,2)$-RAC and $(3,2)$-RAC. The ref. \cite{Ambainis_2024} provides advantageous qubit QC protocols with Bob implementing projective measurements for these contexts. The following results are direct implication of Theorem \ref{th2}. 
\begin{center}
\begin{tabular}{|c|c|c|}
 \hline
 Scenario & $\mathcal{S}$ & $\mathcal{D}_C^\mathcal{S}/\mathcal{D}_{EACC}^\mathcal{S}$ \\ 
 \hline\hline
 (2,2) RAC & $\frac12\!\left(1+\frac{1}{\sqrt{2}}\right)$ & $\geqslant \sqrt{2}$ \\ 
 \hline
 (3,2) RAC & $\frac12\!\left(1+\frac{1}{\sqrt{3}}\right)$ & $\geqslant 2(\sqrt{3}-1)$ \\ 
 \hline
\end{tabular}
\end{center}

Semidefinite programming (SDP) analysis shows that the same advantage can also be achieved in adversarial scenario. In these scenarios, Alice performs projective measurements whose distinguishability is maximized when using a maximally entangled probe. Consequently, the distinguishability cannot increase further, which guarantees at least the same advantage as obtained in the previous case.

 But it would be an interesting venture if we find an advantageous EAQC protocol in the same dimension of EACC. For $(3,2)$-RAC, the ref. \cite{piveteau2022entanglement} gives an EAQC protocol which gives the success probability $1/2+1/\sqrt{6}$. To get this success probability classically, $\mathcal{D}_C\geqslant \frac{3}{\sqrt{6}}-\frac12$. In this protocol, Alice uses eight qubit unitaries corresponding to her eight inputs. So the highest value of distinguishability of these unitaries is $2^2/8=1/2$ (Lemma \ref{lem1}). So the advantage $\frac{\mathcal{D}_C^\mathcal{S}}{\mathcal{D}_{EAQC}^\mathcal{S}}
\geq(\frac{3}{\sqrt{6}}-\frac12)/(\frac12)=(\sqrt{6}-1)$, when $\mathcal{S}=1/2+1/\sqrt{6}$.

 \subsection{Equality Problem defined by Graphs}
The communication task is defined by an arbitrary graph $G$ having $N$ vertices.  Alice and Bob receive input from the vertex set of $G$, i.e., $x,y\in [N]$, sampled from the uniform distribution. Let $G_x$ denote the set of vertices in $G$ that are adjacent to $x$, with $N_x$ representing the number of vertices adjacent to $x$. In the task, there is a promise that $x=y$ or $x \in G_y$. Bob's aim is to differentiate between these two cases, giving the correct output $z=1$ if $x=y$ and $z=2$ if $x\in G_y$. Hence, the success metric,  
\be \label{SG}
\mathcal{S}(G) = \frac{1}{\sum_x N_x+N} \sum_{y=1}^N \left( p(1|x=y,y) + \sum_{x\in G_y} p(2|x,y) \right) .
\ee 
The quantum advantage for this task in terms of the dimension of the communicated systems has been studied in \cite{saha2019,saha2023}.

For any graph $G$, the following holds in classical communication,
\be\label{scg}
\frac{1}{N\alpha(G) } \left( \left(\sum_x N_x + N \right)\left(\mathcal{S}(G) -1 \right) + N \right) \leqslant \mathcal{D}_C ,
\ee 
where $\alpha(G)$ is the independence number of the graph $G$\cite{PRR_24}.

From the ref. \cite{PRR_24} and Theorem \ref{th2} , we can infer the following result.

    Let us consider $N$-cyclic graph (denoted by $\Delta_N$) such that $N\geqslant 5$ and odd. There exist EACC strategies such that 
    \be \label{ratio}
\frac{\mathcal{D}_C^\mathcal{S}(\Delta_N)}{\mathcal{D}_{EACC}^\mathcal{S}(\Delta_N)} = \frac{N}{N-1}\left({1-2\sin^2\left(\frac{\pi}{2N}\right)}\right) > 1,
    \ee 
    for $\mathcal{S} = 1-\frac{2}{3}\sin^2\big(\frac{\pi}{2N}\big)$.

\subsection{Pair Distinguishability Task}
We consider a generalized version of the task introduced in \cite{Chaturvedi2020quantum}. Alice receives input $x \in [N]$ with $p_x=1/N$ and Bob receives a pair of inputs $y \equiv (x,x')$ randomly such that $x,x' \in [N]$ and $x< x'$. The task is to guess $x$. In other words, given Bob's input $(x,x')$, the task is to distinguish between these two inputs. Subsequently, the average success metric, 
\be
\mathcal{S}(N) = \frac{1}{N(N-1)} \sum_{\substack{x,x' \\  x<x' }} \big[ p(x|x,(x,x')) + p(x'|x',(x,x')) \big] .
\ee  
The ref. \cite{PRR_24} provides the following bound in CC:
\be \label{scpd}
(N-1)(\mathcal{S}_C(N) -1) +1 \leqslant  \mathcal{D}_C .
\ee 
The same reference also exhibits the quantum advantageous qubit protocols in quantum communication for $N=3,4,5,6$. By Theorem \ref{th2}, we can achieve same advantage in EACC. Similar to RAC, numerically we checked that same advantage can be gained in adversarial setting.

\subsection{Simplest communication task \cite{Chaturvedi2020quantum}}
The ref. \cite{Chaturvedi2020quantum} gives the bound of classical success probability as the function of distinguishability of this task as following:
\bea
&&p(1|1,1)+p(1|1,2)+p(1|2,1)+p(2|2,2)+p(2|3,1)\nonumber\\
&&\leqslant 2+3\mathcal{D}_C,
\eea
which gives the bound on classical distinguishability as,
\bea
\mathcal{D}_C\geqslant \frac13(\mathcal{S}_C-2).
\eea
The same reference also demonstrates the advantage in qubit quantum communication protocol. According to Theorem \ref{th2}, we conclude that $\mathcal{D}_C^\mathcal{S}/\mathcal{D}_{EACC}^\mathcal{S}\geqslant \frac{1/3(3+\sqrt{2}-2)}{2/3}=\frac{1+\sqrt{2}}{2}$, when $\mathcal{S}=3+\sqrt{2}$.

\section{Class of tasks where non-maximally entangled states yield a better advantage than a maximally entangled state}

In this section, we take a tilted version of $(2,2)$ RAC. In this setting, our aim is to demonstrate that sharing a qubit-qubit non-maximally entangled state can be advantageous compared to a same dimensional maximally entangled state if Alice uses any four two-outcome measurements corresponding to her inputs and relay the outcomes to Bob.

The success merit of this task is defined as following:
\bea\label{bell}
\mathcal{S}(\theta)&=&\sum_{a,b,x,y}(-1)^{a+b+xy}\beta_{ax}p(b|ax,y)\nonumber\\
&& +\alpha\{\beta_{00} p(0|0,0,0)+\beta_{10} p(0|1,0,0)\nonumber\\
&&-\beta_{00} p(1|0,0,0)-\beta_{10} p(1|1,0,0)\},
\eea
where $\beta_{00}=\cos\mu\cos^2\theta-\frac{\cos\mu}{2}+\frac12, \beta_{10}=1-\beta_{00},\beta_{01}=-\sin\mu\cos^2\theta+\frac{\sin\mu}{2}+\frac12,\beta_{11}=1-\beta_{01}$, $\alpha=\frac{2}{\sqrt{1+2\tan^2(2\theta)}}$. The relation between $\mu$ and $\theta$ such that $\cos\mu=\frac{1}{\sqrt{1+\sin^2(2\theta)}}$. 

 Now the right hand side of above equation is in the form of success merit of two party communication task defined at \eqref{SGen}. This communication task consists of Alice's input $x'=(a,x) \in\{00,01,10,11\}$, Bob's input $y\in\{0,1\}$ and Bob's output $b\in\{0,1\}$. Fixing the value of $\theta$ determines the game. First, we describe the method of evaluating the classical bound.

Using the same method of \eqref{scd1}, we can always write any success metric in the encoding strategy. Similarly, the success merit at \eqref{bell} can be written as,
\bea
\mathcal{S}_C(\theta) &\leqslant &  \sum_{m} \max \{(1+\alpha)\beta_{00}p_e(m|00)+\beta_{01}p_e(m|01) \nonumber \\
 && -(1-\alpha)\beta_{10}p_e(m|10)-\beta_{11}p_e(m|11), \nonumber \\
 &&-(1+\alpha)\beta_{00}p_e(m|00)-\beta_{01}p_e(m|01)+ \nonumber \\
 &&(1-\alpha)\beta_{10}p_e(m|10)+\beta_{11}p_e(m|11) \}  \nonumber \\
 && +\sum_{m} max \{ \beta_{00}p_e(m|00)-\beta_{01}p_e(m|01)\nonumber\\
 &&-\beta_{10}p_e(m|10) 
 +\beta_{11}p_e(m|11),\nonumber \\ 
 &&-\beta_{00}p_e(m|00)+\beta_{01}p_e(m|01) \nonumber \\
 &&+\beta_{10}p_e(m|10)-\beta_{11}p_e(m|11)\}.
\eea

The constraint of this communication task is distinguishability of Alice's inputs and we fix the highest value of distinguishability to $1/2$. The equation \eqref{DC} implies that,
\bea
\sum_m\max\{p_e(m|00), p_e(m|01), p_e(m|10), p_e(m|11)\}\leqslant 2;\nonumber\\
\eea
The encoding variables $\{p(m|ax)\}_m$ forms a convex set for each $a,x$, which means  $\{p(m|ax)\}_m$ represents the vertices of a polytope. We can get the maximum value of $\mathcal{S}_C(\theta)$ by numerically computing all the possible values the vertices of the polytope can take under the above constraint. In principal, the number of messages can be anything; but it can be proved that the sufficient number of messages is same as the number of inputs \cite{Tavakoli2020}. For example, we evaluate that 
$
\mathcal{S}_C(\theta=\pi/3)= 2.7559.
$

 Now we illustrate a qubit EACC protocol that is advantageous with respect to CC.
 
 Alice and Bob share a non-maximally entangled state, $\ket{\Psi}=\cos\theta\ket{00}+\sin\theta\ket{11}$, to implement the protocol. For her four possible inputs, Alice performs four measurements on her subsystem of the shared state and communicates the outcome to Bob. Alice’s measurements are as follows:
$
M_{00}=\left\{\ket{\omega}\bra{\omega},\ket{\omega^\perp}\bra{\omega^\perp}\right\},
M_{10}=\left\{\ket{\omega^\perp}\bra{\omega^\perp},\ket{\omega}\bra{\omega}\right\},$
$
M_{01}=\left\{\ket{\tau}\bra{\tau},\ket{\tau^\perp}\bra{\tau^\perp}\right\},$ $
M_{11}=\left\{\ket{\tau^\perp}\bra{\tau^\perp},\ket{\tau}\bra{\tau}\right\},
$
where $\ket{\omega}=\cos\frac{\mu}{2}\ket{0}+\sin\frac{\mu}{2}\ket{1}$ and $\ket{\tau}=\cos(\frac{\mu}{2}+\frac{\pi}{4})\ket{0}+\sin(\frac{\mu}{2}+\frac{\pi}{4})\ket{1}$.
 
Upon receiving the first outcome, Bob performs measurement in the $\sigma_z$ basis for $y=0$, and in the $\sigma_x$ basis for $y=1$,  If Bob receives the second outcome, he applies the same measurements but reverses the outcome (as of Theorem \ref{th2}).

For finding advantage, we choose the alternate definition as we do not have the bound of classical success merit as a function of distinguishability.
With this strategy, the success probability $\mathcal{S}_{EACC}(\theta)=\sqrt{8+2\alpha^2}$. 
In the EACC strategy, the distinguishability of Alice’s measurements is $\mathcal{D}_{EACC}^{\Psi}\leqslant 1/2$ (Using semi-definite programming). Consequently, at $\theta=\pi/3$, the advantage is given by $\mathcal{S}_{EACC}(\theta=\pi/3) / \mathcal{S}_C(\theta=\pi/3) =3.0237/2.7559=1.0972$ at $\mathcal{D}_C\le1/2,\mathcal{D}_{EACC}^{\Psi}\le 1/2$.  
\subsection*{Non-optimality of maximally entangled state}
In this subsection, we prove that maximally entangled state is non-optimal for a class of tilted RAC. We choose that the distinguishability of Alice's measurements must not exceed $1/2$.
Let us take the parties share maximally entangled state as the pre-shared resource.
\subsubsection*{Optimality of Bob's strategy}
In general, Alice can employ four two-outcome measurements (without assuming that two of them are identical up to a reversal of outcome ordering, as considered above). Suppose that, depending on Alice’s outcome, the reduced states on Bob’s side are $\rho_{ax}$ and $\bar{\rho}_{ax}$. For any measurements, the following holds:
\be\label{rel}
q_{ax}\rho_{ax}+(1-q_{ax})\bar{\rho}_{ax}=\frac{\I}{2},
\ee
where $q_{ax}$ is the probability of getting $\rho_{ax}$ at Bob's side.
Based on Alice’s message and his own input, Bob can choose an appropriate set of measurements. This implies that a sufficient number of Bob’s measurements is four, which we denote by $\{N^{my}_b\}_{m=0}^1$ for each $y,b$. Thus, the correlation can be written as  
\bea
p(b|a,x,y)&=&q_{ax}\tr(\rho_{ax}N^{0y}_b)+(1-q_{ax})\tr(\Bar{\rho}_{ax}N^{1y}_b)\nonumber\\
&=& q_{ax}\left(\tr(\rho_{ax}N^{0y}_b)-\tr(\rho_{ax}N^{1y}_b)\right)+\frac12\tr(N^{1y}_b).\nonumber\\
\eea
In the second line, we have used the relation \eqref{rel}. Similarly, one can derive the expression for $\mathcal{S}(\theta)$ of \eqref{bell}, which takes the form  
\begin{widetext}
\bea\label{S_m}
&&\mathcal{S}_{\phi^+}(\theta)\nonumber\\
&=& \tr\left[N^{00}_0\{(-1)^{a}q_{ax}\beta_{ax}\rho_{ax}+\alpha(q_{00}\beta_{00}\rho_{00}+q_{10}\beta_{10}\rho_{10})\}\right]
- \tr\left[N^{10}_0\{(-1)^{a}q_{ax}\beta_{ax}\rho_{ax}+\alpha(q_{00}\beta_{00}\rho_{00}+q_{10}\beta_{10}\rho_{10})\}\right]\nonumber\\
&& + \tr\left[N^{00}_1\{(-1)^{a+1}q_{ax}\beta_{ax}\rho_{ax}-\alpha(q_{00}\beta_{00}\rho_{00}+q_{10}\beta_{10}\rho_{10})\}\right]- \tr\left[N^{10}_1\{(-1)^{a+1}q_{ax}\beta_{ax}\rho_{ax}-\alpha(q_{00}\beta_{00}\rho_{00}+q_{10}\beta_{10}\rho_{10})\}\right]\nonumber\\
&&+ \tr\left[N^{01}_0\{(-1)^{a+x}q_{ax}\beta_{ax}\rho_{ax}\}\right]- \tr\left[N^{11}_0\{(-1)^{a+x}q_{ax}\beta_{ax}\rho_{ax}\}\right]+\tr\left[N^{01}_1\{(-1)^{a+1+x}q_{ax}\beta_{ax}\rho_{ax}\}\right]\nonumber\\
&& - \tr\left[N^{11}_1\{(-1)^{a+1+x}q_{ax}\beta_{ax}\rho_{ax}\}\right]+ \{(-1)^{a+x}\beta_{ax}\}\tr(N^{11}_0-N^{11}_1)+ \{(-1)^{a}\beta_{ax}+\alpha(\beta_{00}+\beta_{10})\}\tr(N^{10}_0-N^{10}_1).
\eea
    
\end{widetext}
Now consider the first two terms. To maximize $\mathcal{S}_M(\theta)$, $N^{00}_0$ must be the projector onto the eigenvector corresponding to the largest eigenvalue of $\{(-1)^{a}q_{ax}\beta_{ax}\rho_{ax}+\alpha(q_{00}\beta_{00}\rho_{00}+q_{10}\beta_{10}\rho_{10})\}$. Similarly, $N^{10}_0$ corresponds to the eigenvector associated with the largest eigenvalue of $-\{(-1)^{a}q_{ax}\beta_{ax}\rho_{ax}+\alpha(q_{00}\beta_{00}\rho_{00}+q_{10}\beta_{10}\rho_{10})\}$. Since the operator inside the brackets acts on a two-dimensional space, it follows that $N^{00}_0$ is orthogonal to $N^{10}_0$. 
As Bob’s measurements have two outcomes, they must be projective. Hence, we can write  
$N^{00}_0=\ket{\chi}\bra{\chi}$ and $N^{10}_0=\ket{\chi^\perp}\bra{\chi^\perp}$.  
Projectivity further implies that $N^{00}_1=\ket{\chi^\perp}\bra{\chi^\perp}$ and, by the same reasoning, $N^{10}_1=\ket{\chi}\bra{\chi}$.  

Applying an analogous argument, if $N^{01}_0=\ket{\kappa}\bra{\kappa}$, then $N^{11}_0=\ket{\kappa^\perp}\bra{\kappa^\perp}$. Consequently, $N^{01}_1=\ket{\kappa^\perp}\bra{\kappa^\perp}$ and $N^{11}_1=\ket{\kappa}\bra{\kappa}$.
As all the measurement operators are projectors, the last two terms of \eqref{S_m} reduces to zero.

That proves that Bob’s strategy—reversing the outcome while using fixed measurements for both of Alice’s messages—is optimal.
\subsubsection*{Optimal strategy of Alice}
We can write \eqref{S_m} as follows:
\begin{widetext}
\bea
\mathcal{S}_{\phi^+}(\theta)
&=& 2q_{00}\beta_{00}\tr\left[\rho_{00}\left\{(1+\alpha)\left(\ket{\chi}\bra{\chi}-\ket{\chi^\perp}\bra{\chi^\perp}\right)+\ket{\kappa}\bra{\kappa}-\ket{\kappa^\perp}\bra{\kappa^\perp}\right\}\right]\nonumber\\
&&+ 2q_{10}\beta_{10}\tr\left[\rho_{10}\left\{-(1-\alpha)\left(\ket{\chi}\bra{\chi}-\ket{\chi^\perp}\bra{\chi^\perp}\right)-\ket{\kappa}\bra{\kappa}+\ket{\kappa^\perp}\bra{\kappa^\perp}\right\}\right]\nonumber\\
&&+2q_{01}\beta_{01}\tr\left[\rho_{01}\left\{\ket{\chi}\bra{\chi}-\ket{\chi^\perp}\bra{\chi^\perp}-\ket{\kappa}\bra{\kappa}+\ket{\kappa^\perp}\bra{\kappa^\perp}\right\}\right]\nonumber\\
&&+2q_{11}\beta_{11}\tr\left[\rho_{11}\left\{-\ket{\chi}\bra{\chi}+\ket{\chi^\perp}\bra{\chi^\perp}+\ket{\kappa}\bra{\kappa}-\ket{\kappa^\perp}\bra{\kappa^\perp}\right\}\right]
\eea  

\end{widetext}
In the above expression, we put the explicit form of Bob's measurements.
The best states of Alice would be the largest eigenvalues of the corresponding quantities of second brackets. The optimum $\rho_{ax}$ are the eigenvectors to the highest eigenvalues of corresponding terms. All the eigenvectors in this case are pure states. As the inputs span $2$-dimensional space, Alice must choose projective measurements. That implies $q_{ax}=1/2,\forall a,x$. Thus we can calculate the highest eigenvalues of each quantities and write the optimum success merit as,
\bea\label{opt}
\mathcal{S}_{\phi^+}(\theta)
&=& \beta_{00}\sqrt{\alpha^2+4(\alpha+1)|\la \chi|\kappa\ra|^2}\nonumber\\
&+& \beta_{10}\sqrt{\alpha^2+4(1-\alpha)|\la \chi|\kappa\ra|^2}\nonumber\\
&+&\beta_{01}2|\la\chi^\perp|\kappa \ra|+\beta_{11}2|\la\chi^\perp|\kappa \ra|.
\eea

Without changing the value of above quantity, we can always choose a unitary $U$ such that $U\ket{\chi}=\ket{0}$ and $U\ket{\kappa}=\ket{\kappa'}=\cos\frac{\eta}{2}\ket{0}+\sin\frac{\eta}{2}\ket{1}$. Therefore, \eqref{opt} reads as,

\bea\label{opt1}
\mathcal{S}_{\phi^+}(\theta)
&=& \beta_{00}\sqrt{\alpha^2+4(\alpha+1)\cos^2\frac{\eta}{2}}\nonumber\\
&+& \beta_{10}\sqrt{\alpha^2+4(1-\alpha)\cos^2\frac{\eta}{2}}\nonumber\\
&+&2|\sin\frac{\eta}{2}|.
\eea 
For the last term, we use the fact that $\beta_{01}+\beta_{11}=1$. For a given task, i.e., when the value of $\theta$ is fixed, one can always determine the maximum value of the success metric by optimizing the value of $\eta$ as the shared resource. If this value is less than $\sqrt{8+2\alpha^2}$, it implies that a non-maximally entangled state achieves a higher success probability than the maximally entangled state. For $\theta=\pi/3$, we can calculate $\mathcal{S}_{\phi^+}(\theta=\pi/3) =2.8298$, where as we already mention $\mathcal{S}_{EACC}(\theta=\pi/3)=3.0237$. Since Alice executes four projective measurements and the highest distinguishability of her four measurements is $1/2$ (By Theorem \ref{th1}), we conclude that a non-maximally entangled state can outperform a maximally entangled state as the shared resource.

For general $\theta$, we provide the following plot at Figure \ref{fig1} to demonstrate the advantage of non-maximally entangle probe.

\begin{figure}[h!]
    \centering \includegraphics[width=1.08\linewidth]{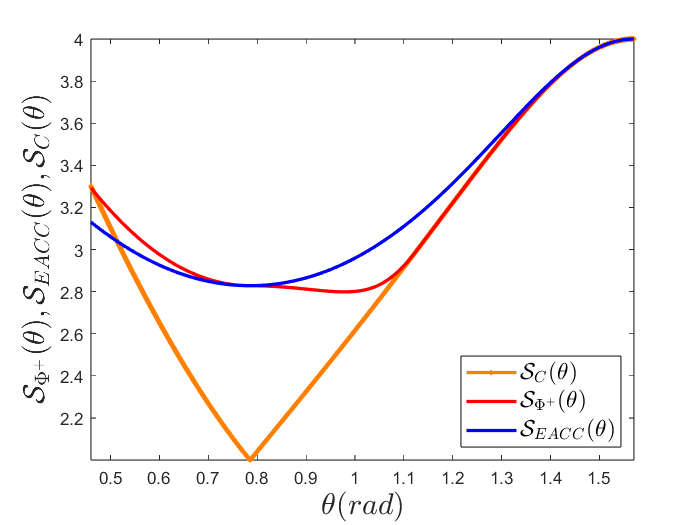}
    \caption{The plot illustrates the advantage of EACC over CC in the tilted $(2,2)$-RAC. The independent axis represents the parameter $\theta$ which is the tilted parameter of success merit, while the dependent axis shows the corresponding values of $\mathcal{S}_{\phi^+}(\theta), \mathcal{S}_{EACC}(\theta)$ and $\mathcal{S}_C(\theta)$. The plot clearly demonstrates that $\mathcal{S}_{EACC}(\theta) > \mathcal{S}_C(\theta)$ for the  $0.79$ rad$\leq\theta\leq$ 1.56 $rad$. It is interesting to see that for some region, maximally entangled state provides a greater advantage over the non-maximally entangled state. But we did not optimize over the strategies with non-maximally entangled states. Therefore, we can not comment about the superiority of maximally entangled state in those values of $\theta$.
}
    \label{fig1}
\end{figure}

\section{Conclusion}
In this work, we present a novel and systematic framework to investigate entanglement-assisted two-party communication tasks under constraints on the distinguishability of the sender’s inputs. Within this framework, we identify several communication tasks that demonstrate an advantage of both entanglement-assisted classical and quantum communication over classical communication supplemented with shared randomness. 
Traditional studies of communication complexity typically characterize such advantages under constraints on the dimension of the communicated message. In contrast, our distinguishability-based measure is particularly relevant in scenarios where preserving the privacy of the sender’s inputs is a primary concern.
Our main results establish an equivalence between quantum communication, entanglement-assisted classical communication, and entanglement-assisted quantum communication protocols. This equivalence further motivates us to introduce a new class of protocols in which both the distinguishability of inputs and the dimension of the communicated message are simultaneously constrained. Within this refined setting, we demonstrate an advantage for entanglement-assisted protocols even in the absence of any input for the receiver. This, in turn, highlights the superiority of entanglement-assisted classical and quantum communication over standard quantum communication. 
Finally, we identify a communication task in which a non-maximally entangled state outperforms a maximally entangled state as a pre-shared resource between the communicating parties.

An intriguing aspect of our study is that the distinguishability of general quantum processes, such as measurements \cite{Manna_111,Datta_2021,ziman} and channels \cite{acin2,piani,watrous,npj,manna2025,manna2026,anti_unit}, plays a central role in entanglement-assisted classical and quantum communication.

 Another important feature of our work arises from the alternative definition of distinguishability, which effectively leads to a collapse of communication complexity. This phenomenon is well known in the context of non-local games~\cite{Botteron,botteron2024,pierre,skrzypczyk,brassard,karol,Grudka}. A similar line of investigation can be carried out for non-local games under distinguishability constraints. Since, in most cases, the distinguishability of quantum channels does not admit a closed-form expression, developing tighter and more efficient numerical methods will be essential for advancing this line of research.  
In most of our results, we demonstrated the existence of advantageous protocols, although the achieved advantage is not optimal. It would therefore be worthwhile to formulate these tasks as semi-definite programs under explicit distinguishability constraints in order to characterize optimal performance. Related dimension-constrained formulations have been studied previously \cite{Vieira_2023,armin}, and extending those techniques to the present setting could yield deeper insights.  
Distinguishability is also highly relevant in the study of ontological models. In this context, the role of state discrimination has been extensively investigated \cite{leifer2014,barrett14,ray2025,chaturvedi2021,epistemic}; however, channel discrimination may play a similar role. Communication tasks of the type considered here could therefore provide a useful framework for exploring such questions. Entanglement in Bell-type settings has already been shown to be useful in several aspects of cryptography \cite{Ekert,bod15}. The framework considered here may therefore suggest new directions for the development of cryptographic protocols, particularly in light of the connection between distinguishability and the privacy of the sender's inputs discussed earlier.

\subsection*{Acknowledgement} This work is supported by STARS (STARS/STARS-2/2023-0809), Govt. of India. AP thanks UGC, India for the Junior Research Fellowship (reference no.: 221610057949)

\bibliography{ref}
\end{document}